%% file: 00-paperVLDB.tex
\documentclass[sigconf, nonacm]{acmart}
\pdfoutput=1

\newcommand\vldbdoi{XX.XX/XXX.XX}
\newcommand\vldbpages{XXX-XXX}
\newcommand\vldbvolume{14}
\newcommand\vldbissue{1}
\newcommand\vldbyear{2020}
\newcommand\vldbauthors{\authors}
\newcommand\vldbtitle{\shorttitle} 
\newcommand\vldbavailabilityurl{URL_TO_YOUR_ARTIFACTS}
\newcommand\vldbpagestyle{plain} 

\usepackage{bibentry}
\usepackage{booktabs}
\usepackage{dcolumn}
\usepackage{amsthm}
\usepackage{xspace}
\usepackage{stmaryrd}
\usepackage{bbm}
\usepackage{yfonts}
\usepackage{mathrsfs}
\usepackage{rotating}
\usepackage{pdflscape}
\usepackage{pgfplots}
\usepackage[ruled,vlined,scleft,linesnumbered]{algorithm2e}
\SetKwComment{Comment}{\texttt{//}\ }{}  

\usepackage{pgf}
\usepackage{tikz}
\usetikzlibrary{shapes}
\usetikzlibrary[graphs]

\usetikzlibrary{trees}

\usetikzlibrary{arrows}
\usetikzlibrary{arrows.meta}
\usetikzlibrary{arrows,automata}
\usetikzlibrary[shapes.geometric]
\usetikzlibrary[calc]
\usetikzlibrary[graphs]
\usetikzlibrary[shapes.multipart]
\usetikzlibrary[topaths]
\usetikzlibrary{positioning}

\usetikzlibrary{backgrounds,matrix,fit}
\usetikzlibrary{babel}
\usetikzlibrary{arrows.meta,chains,decorations.pathreplacing}

\usetikzlibrary{intersections}
\usetikzlibrary{shadows}
\usetikzlibrary{patterns}

\pgfplotsset{compat=1.16}

\DeclareMathAlphabet{\mathpzc}{OT1}{pzc}{m}{it}

\newcommand{\rank}{\mathsf{rank}}
\newcommand{\select}{\mathsf{select}}
\renewcommand{\next}{\mathsf{next}}

\newcommand{\predecessor}{\mathsf{predecessor}}
\newcommand{\successor}{\mathsf{successor}}

\newcommand{\set}{\ensuremath{S}\xspace}
\newcommand{\univSize}{u}
\newcommand{\setSize}{n}
\newcommand{\wordSize}{w}

\newcommand{\true}{\mathsf{true}}
\newcommand{\false}{\mathsf{false}}
\newcommand{\gap}[1]{\mathtt{gap}(#1)}

\newcommand{\lowerB}{\mathcal{B}(\setSize, \univSize)}
\newcommand{\rle}[1]{\mathtt{rle}(#1)}

\renewcommand{\log}{\lg}

\newcommand{\bit}[1]{\ensuremath{\mathbf{#1}}}

\newcommand{\bigoh}[1]{\ensuremath{\mathrm{O}(#1)}}

\newcommand{\PT}{\ensuremath{\mathtt{PT}}\xspace}

\newcommand{\Cn}{\ensuremath{C^{(n)}}}
\newcommand{\cbv}{C_{\!\set}}

\newcommand{\hzrun}[1]{\ensuremath{H_0^{\mathrm{run}}(#1)}}
\newcommand{\hzgap}[1]{\ensuremath{H_0^{\mathrm{gap}}(#1)}}

\newcommand{\bintrie}[1]{\ensuremath{\mathsf{bintrie}(#1)}}
\newcommand{\certificate}[1]{\ensuremath{\mathsf{cert}(#1)}}
\newcommand{\trie}[1]{\ensuremath{\mathtt{trie}(#1)}}

\newcommand{\trieRun}[1]{\ensuremath{\mathtt{rTrie}(#1)}}
\newcommand{\comentar}[1]{}

\newcommand{\patrascu}{{P}\v{a}tra\c{s}cu\xspace}

\usepackage{balance}
\usepackage[all]{nowidow}

\begin{document}
\title{Trie-Compressed Intersectable Sets}

\author{Diego Arroyuelo}
\orcid{0000-0002-2509-8097}
\affiliation{%
  \institution{Dept.~of Informatics, Universidad T\'ecnica Federico Santa \\ Mar\'ia, Chile, \& Millennium Institute for Foundational Research on Data, Chile}
  \streetaddress{Vicu\~na Mackenna 3939}
}
\email{darroyue@inf.utfsm.cl}

\author{Juan Pablo Castillo}
\affiliation{%
  \institution{Dept.~of Informatics, Universidad T\'ecnica Federico Santa \\ Mar\'ia, Chile, \& Millennium Institute for Foundational Research on Data, Chile}
}
\email{juan.castillog@sansano.usm.cl}

\begin{abstract}
We introduce space- and time-efficient algorithms and data structures for the offline set intersection problem. 
We show that a sorted integer set $S \subseteq [0{..}u)$ of $n$ elements can be represented using compressed space while supporting $k$-way intersections in adaptive $\bigoh{k\delta\lg{\!(u/\delta)}}$ time, $\delta$ being the alternation measure introduced by Barbay and Kenyon. Our experimental results suggest that our approaches are competitive in practice, outperforming the most efficient alternatives (Partitioned Elias-Fano indexes, Roaring Bitmaps, and Recursive Universe Partitioning (RUP)) in several scenarios, offering in general relevant space-time trade-offs.  
\end{abstract}

\maketitle

\pagestyle{\vldbpagestyle}
\begingroup\small\noindent\raggedright\textbf{PVLDB Reference Format:}\\
\vldbauthors. \vldbtitle. PVLDB, \vldbvolume(\vldbissue): \vldbpages, \vldbyear.\\
\href{https://doi.org/\vldbdoi}{doi:\vldbdoi}
\endgroup
\begingroup
\renewcommand\thefootnote{}\footnote{\noindent
This work is licensed under the Creative Commons BY-NC-ND 4.0 International License. Visit \url{https://creativecommons.org/licenses/by-nc-nd/4.0/} to view a copy of this license. For any use beyond those covered by this license, obtain permission by emailing \href{mailto:info@vldb.org}{info@vldb.org}. Copyright is held by the owner/author(s). Publication rights licensed to the VLDB Endowment. \\
\raggedright Proceedings of the VLDB Endowment, Vol. \vldbvolume, No. \vldbissue\ %
ISSN 2150-8097. \\
\href{https://doi.org/\vldbdoi}{doi:\vldbdoi} \\
}\addtocounter{footnote}{-1}\endgroup

\ifdefempty{\vldbavailabilityurl}{}{
\vspace{.3cm}
\begingroup\small\noindent\raggedright\textbf{PVLDB Artifact Availability:}\\
The source code, data, and/or other artifacts have been made available at \url{https://github.com/jpcastillog/compressed-binary-tries}.
\endgroup
}

\input{01-intro}

\input{02-preliminary}

\input{03-compressed-tries}

\input{04-implementation}

\input{05-experimental-results}

\input{06-conclusions}

\begin{acks}
This work was funded by ANID – Millennium Science Initiative Program – Code ICN$17\_002$, Chile (both authors). We thank Gonzalo Navarro, Adri\'an G\'omez-Brand\'on, and Francesco Tosoni for enlightening  comments, suggestions, and discussions about this work.
\end{acks}


\balance

\bibliographystyle{ACM-Reference-Format}

\end{document}

%% file: 01-intro.tex
\section{Introduction}
\label{section:introduction}

\emph{Sets} are one of the most fundamental mathematical concepts related to the storage of data. Operations such as intersections, unions, and differences are key for querying them. For instance, the use of logical \texttt{AND} and \texttt{OR} operators in web search engines translate into intersections and unions, respectively. Representing sets to support their basic operations efficiently has been a major concern since many decades ago \cite{AHU74}. In several applications, such as query processing in information retrieval \cite{BCC2010} and database management systems \cite{ENbook11}, sets are known in advance to queries, hence data structures can be built to speed up query processing. 
With this motivation, in this paper we focus on the following problem.
\begin{definition}[The \emph{Offline Set Intersection Problem}, \textbf{OSIP}] 
Let $\mathcal{S} = \{S_1,\ldots , S_N\}$ be a family of $N$ integer sets, each of size $|S_i| = \setSize_i$ and universe $[0{..}\univSize)$. The OSIP consist in preprocessing this family to support query instances of the form $\mathcal{Q} = \{i_1, \ldots, i_k\} \subseteq [1{..}N]$,
which ask to compute $\mathcal{I}(\mathcal{Q}) = \bigcap_{i_j\in \mathcal{Q}}{S_{i_j}}$.
\end{definition}
Typical applications of this problem include the efficient support of join operations in databases \cite{ENbook11,Vicdt14}, query processing using inverted indexes in Information Retrieval (IR) \cite{BCC2010,MG}, and computational biology \cite{/LQipeee17}, among others.
Building a data structure to speed up intersections, however, increases the space usage.
Today, data-intensive applications compel not only for time-efficient solutions: space-efficient solutions allow one to maintain data structures of very large datasets entirely in main memory, thus avoiding accesses to slower secondary storage. Compact, succinct, and compressed data structures have been the solution to this problem in the last few decades \cite{Navarro16}. We study here compressed data structures to efficiently support the OSIP.

We assume the word RAM model of computation with word size 
$\wordSize = \Theta(\lg{\univSize})$. Arithmetic, logic, and bitwise operations, as well as accesses to $w$-bit memory cells, take $\bigoh{1}$ time.

The set intersection problem has been studied in depth, resulting in a plethora of algorithms and data structures that we briefly review in what follows. First, let us look at the online version of the problem, where sets to be intersected are given at query time, not before, so there is no time to preprocess them. Algorithms like the ones by Baeza-Yates \cite{BYcpm04}, Demaine, Lopez-Ortiz, and Munro \cite{DLMsoda00}, and Barbay and Kenyon \cite{BKtalg08} are among the most efficient approaches. In particular, the two latter algorithms are adaptive, meaning that they are able to perform faster on ``easier'' query instances. In particular, Barbay and Kenyon's algorithm runs in optimal $\bigoh{\delta\sum_{i=1}^{k}{\lg{\!(n_i/\delta)}}}$, where $\delta$ is the so-called alternation measure, which quantifies the query difficulty. The algorithm by Demaine et al.~\cite{DLMsoda00} has running time $\bigoh{k\delta\lg{\!(n/\delta)}}$ (where $n = \sum_{i\in\mathcal{Q}}{n_i}$), which is optimal when $\max_{i\in \mathcal{Q}}{\{\lg{n_i}\}} = \bigoh{\min_{i\in\mathcal{Q}}{\{\lg{n_i}\}}}$ \cite{BKsoda02}. These algorithms require sets to be stored in plain form, in a sorted array of $\Theta(n_i\lg{u})$ bits, $i=1,\ldots,N$. This can be excessive when dealing with large databases. An interesting question is whether one can support the OSIP in time proportional to $\delta$, while using compressed space. 

On the offline version of the problem, we can cite the vast literature on inverted indexes \cite{MG,ZMcsur06,BCC2010,PVcsur21}, where the main focus is on practical space-efficient set representations supporting intersections. Although there are highly-efficient approaches in practice within these lines, it is hard to find worst-case guarantees both in space usage and query time. One notorious exception is the data structure by Ding and Konig \cite{DKpvldb11}, which computes intersection in $\bigoh{n/\sqrt{w} + k|\mathcal{I(Q)}|}$ time, and uses linear space. According to Ding and Konig's experiments, the space can be improved to use bout 1.88 times the space of an Elias $\gamma$/$\delta$ compressed inverted index \cite{DKpvldb11}.

\comentar{There are papers studying the offline intersection problem, as for instance:
\begin{itemize}
\item Cohen and Porat \cite{CPtcs10}: segun el paper de Chen and Shen, el espacio es
$\bigoh{|T| |D|}$, pero no me queda claro (y no lo dicen) cual es el tamaño
$|T|$ del árbol que construyen. Tampoco queda claro que es $|D|$: es la cantidad
de conjuntos en la base de datos, es decir, $N$? O es el tamaño del universo, $u$? 

\item Chen and Shen \cite{CStcs16}: en el peor caso no pueden mejorar demasiado
respecto a la cota inferior del modelo de comparaciones. 
Sin embargo, puede ser bastante mejor que eso. En definitiva, ese paper muestra
que con las estructuras adecuadas se pueden lograr mejores tiempos que el modelo
de comparaciones, pero en peor caso llegan justo a eso.
Las cotas que dan no son muy precisas, mas alla de que tienen cota superior de peor caso.
Pero solo pueden decir que en general es mucho mejor que el peor caso. Yo con lo mio
puedo precisar mucho mas esas cotas.
Su metodo boils down a una interseccion online sobre un conjunto de intervalos que, sí,
pueden ser bastante mas pequeños que los conjuntos originales. Pero en el peor caso no
lo son. Tambien estudian una manera de usar estructuras de datos de tipo LCA para reducir
el tiempo, pero siguen teniendo mal peor caso. 
\end{itemize}
}

In this paper we revisit the classical set representation and intersection algorithm by Trabb-Pardo \cite{TrabbPardo}. To represent a set $S$, Trabb-Pardo inserts the $w$-bit binary codes of each $x\in S$ into a binary trie data structure \cite{Fcacm60}. Intersection on this representation can be computed traversing all the tries corresponding to the query sets in coordination. We will show that an approach like this allows for a space- and time-efficient implementation, with theoretical guarantees and competitive practical performance. Next, we enumerate our contributions.


\subsection{Contributions}

The main contributions of this paper are as follows:
\begin{itemize}
    \item We introduce the concept of \emph{trie certificate} of an instance of the intersection problem. This is analogous to the proofs by Demaine et al.~\cite{DLMsoda00} and partition certificates by Barbay and Kenyon \cite{BKtalg08}, which were introduced to analyze their adaptive intersection algorithms. See Definition \ref{def:trie-certificate}.
    
    \item We use trie certificates to show that the classical intersection algorithm by Trabb-Pardo \cite{TrabbPardo} is actually an adaptive intersection algorithm. In particular, we show that its running time to compute $\mathcal{I(Q)}$ is $\bigoh{k\delta\lg{(u/\delta)}}$, where $\delta$ is the alternation measure by Barbay and Kenyon \cite{BKtalg08}. See Theorem \ref{theor:TP-adaptive}
    
    \item We introduce compact trie data structures that support the navigation operations needed to implement the intersection algorithm. We then show that these compact data structures yield a compressed representations of the sets, bounding their size by the trie entropy from Klein and Shapira \cite{KSdcc02} (and a variant we introduce in Definition \ref{def:rTrie}). See Theorems \ref{theor:compressed-TP} and \ref{theor:run-compressed-TP}.

    \item We show preliminary experimental results that indicate that our approaches are appealing not only in theory, but also in practice, outperforming the most competitive state-of-the-art approaches in several scenarios. For instance, we show that our algorithm computes intersections 1.55--3.32 times faster than highly-competitive Partitioned Elias-Fano indexes \cite{OVsigir14}, using 1.15--1.72 times their space. We also compare with Roaring Bitmaps \cite{LKKDOSKspe18}, being 1.07--1.91 times faster while using about 0.48--1.01 times their space. Finally, we compared to RUP \cite{Pdcc21}, being 1.04--2.18 times faster while using 0.82--1.19 times its space. 
\end{itemize}
Overall, our work seems to be a step forward in bridging the gap between theory and practice in this important line of research. 



%% file: 02-preliminary.tex
\section{Preliminaries and Related Work} \label{sec:previous-work}

\subsection{Operations of Interest}

Besides intersections (and other set operations like unions and differences), there are several other operations one would want to support on an integer set. The following are among the most fundamental operations:
\begin{itemize}
    \item $\rank(S, x)$: for $x \in [0{..}u)$, yields $|\{y \in S,~y\le x\}|$.
    \item $\select(S, j)$: for $1\le j\le |S|$, yields $x \in S$ s.t.~$\rank(S,x)=j$.
    \item $\predecessor(S, x)$: for $x \in [0{..}u)$, yields $\max{\{y \in S,~y\le x\}}$.
    \item $\successor(S, x)$: for $x \in [0{..}u)$, yields $\min{\{y\in S,~y \ge x\}}$.
\end{itemize}
A set $S$ can be also described using its \emph{characteristic bit vector} (cbv, for short) $\cbv[0{..}u)$, such that $\cbv[x_i] = \bit{1}$ if $x_i \in S$, $\cbv[x_i] = \bit{0}$ otherwise. On the cbv $\cbv$ we define the following operations:
\begin{itemize}
    \item $\cbv.\rank_\bit{1}(x)$: for $x\in [0{..}u)$,  yields the number of $\bit{1}$s in $\cbv[0{..}x]$.
    \item $\cbv.\select_{\bit{1}}(k)$: for $1\le j\le |S|$, yields the position $0\le x <u$ s.t.~$\cbv.\rank(x)=j$.
\end{itemize}
Notice that $\rank(S, x) \equiv \cbv.\rank_{\bit{1}}(x)$ and $\select(S, j) \equiv \cbv.\select_{\bit{1}}(j)$.

\subsection{Set Compression Measures}

A \emph{compression measure} (or \emph{entropy}) quantifies the amount of bits needed to encode data using a particular compression model. 
For an integer universe $U = [0{..}u)$, let $\Cn\subseteq 2^{U}$ denote the class of all sets $S \subseteq U$ such that $|S| = n$. In this section, we assume $S=\{x_1,\ldots, x_n\}$, for $0\le x_1 \le \cdots \le x_n < \univSize$. We review next typical compression measures for integer sets $S \in \Cn$, that will be of interest for our work. For set $S$, let $\cbv [0..\univSize)$ denote its \emph{characteristic bit vector} (cbv for short), such that $\cbv [x] = \bit{1}$ iff $x\in \set$, $\cbv[x] = \bit{0}$ otherwise.

\subsubsection{The Information-Theoretic Worst-Case Lower Bound.}
A worst-case lower bound on the number of bits needed to represent any $S\in \Cn$ can be
obtained using the following information-theoretic argument. As $|\Cn| = \binom{u}{n}$, at 
least $\mathcal{B}(n, u)  =  \lceil \lg{\binom{u}{n}}\rceil$ bits are needed to differentiate a given $S\in \Cn$ among all other possible sets. This bound is tight, as there are set representations achieving this bound \cite{RRR07,Patrascu08}.
If $n \ll u$, $ \mathcal{B}(n, u)\approx n\lg{e} + n\lg{\frac{u}{n}} - \bigoh{\lg{u}}$ bits (using the Stirling approximation of $n!$). This is just a worst-case lower bound: some sets in $\Cn$ can be represented using less bits, as we will see next. 

\subsubsection{The $\gap{S}$ Compression Measure.}
A well-studied compression measure is $\gap{S}$, defined as follows.  Let us denote $g_{1} = x_{1}$ and, for $i = 2, \ldots, n$, define $g_{i} = x_{i} - x_{i-1}-1$. Then, we define the $\gap{S}$ measure as follows:
$$\gap{\set} = \sum_{i=1}^{\setSize}{\lfloor \lg{g_i}\rfloor + 1}.$$
Note that $\gap{S}$ is the amount of bits required to represent $S$ provided we encode the sequence of gaps $\mathcal{G} = \langle g_1, \ldots, g_n\rangle$, using $\lfloor \lg{g_{i}}\rfloor+1$ bits per gap.
This measure exploits the variation in the gaps between consecutive set elements: the closer the elements, the smallest this measure is.
The following lemma upper bounds $\gap{S}$:
\begin{lemma} \label{lemma:gap-vs-wc-lower-bound}
Given a set $S\subseteq [0{..}u)$ of $n$ elements, it holds that $\gap{S} \le \mathcal{B}(n, u)$, with equality only when $g_{i} = \frac{u}{n}$ (for $i = 1,\ldots,n$). 
\end{lemma}
The $\gap{S}$ measure has been traditionally used in applications like inverted-index compression in information retrieval \cite{BCC2010} and databases \cite{MG}.

\subsubsection{The $\rle{S}$ Compression Measure.}
Run-length encoding is a more appropriate model when set elements tend to be clustered into runs of consecutive elements.
Following Arroyuelo and Raman \cite{ARalg22}, 
a \emph{maximal run of successive elements} $R \subseteq \set$ contains $|R| \ge 1$ elements $x_i,x_i+1,\ldots,x_{i}+|R|-1$, such that $x_{i}-1 \not \in \set$ and
$x_{i} + |R| \not \in \set$. Let $R_1,\ldots, R_r$ be the partition of set $\set$ into maximal runs of successive elements, 
such that $\forall x\in R_i, \forall y\in R_j$, $x < y \iff i < j$. 
Let $z_1, \ldots , z_r$ be defined as $z_1 = \min{\{R_1\}}$,
and $z_i = \min{\{R_i\}} - \max{\{R_{i-1}\}} - 1$, and $\ell_1, \ldots, \ell_r$ be such that $\ell_i = |R_i|$, for $i=1,\ldots, r$.
According to these definitions, the cbv of $S$ is denoted as $\cbv [0..\univSize) = \bit{0}^{z_1}\bit{1}^{\ell_1}\bit{0}^{z_2}\bit{1}^{\ell_2}\cdots
\bit{0}^{z_r}\bit{1}^{\ell_r}$, and the set can be represented through the sequences $\mathcal{Z} = \langle z_1, \ldots, z_r\rangle$ and $\mathcal{O} = \langle \ell_1, \ldots, \ell_r\rangle$ of
$2r$ lengths of the alternating 0/1-runs in $\cbv$ (assume wlog that $\cbv$ begins with $\bit{0}$ and ends with $\bit{1}$).  Then:

$$
\rle S = \sum_{i=1}^{r}{\left(\lfloor \log{(z_i-1)} \rfloor + 1\right)} + \sum_{i=1}^{r}{\left(\lfloor \log{(\ell_i-1)} \rfloor + 1\right)}.
$$

\begin{lemma}[\cite{FGGVtalg06}]
Given a set $S\subseteq [0{..}u)$ of $n < \univSize/2$ elements, it holds that $\rle{S} < \lowerB + \setSize + \bigoh{1}$.  
\end{lemma}

\subsubsection{The $\trie{S}$ Compression Measure.} \label{sec:trie-measure}
Let us consider now representing a set $S\in \Cn$ using a binary trie denoted \bintrie{S}, where the $\ell = \lceil\lg{u}\rceil$-bit binary encoding of every element is added. Hence, \bintrie{S} has $|S|$ external nodes, all at depth $\ell$. For an external trie node $v$ corresponding to an element $x_i \in S$, the root-to-$v$ path is hence labeled with the binary encoding of $x_i$. This approach has been used for representing sets since at least the late 70s by Trabb-Pardo \cite{TrabbPardo}, former Knuth's student. Consider the example sets $S_1 = \{1,3,7,8,9,11,12\}$, and $S_2 = \{2,5,7,12,15\}$ over universe $[0{..}16)$, that we shall use as running examples. Figure \ref{fig:binary tries} shows the corresponding tries \bintrie{S_1} and \bintrie{S_2}. 
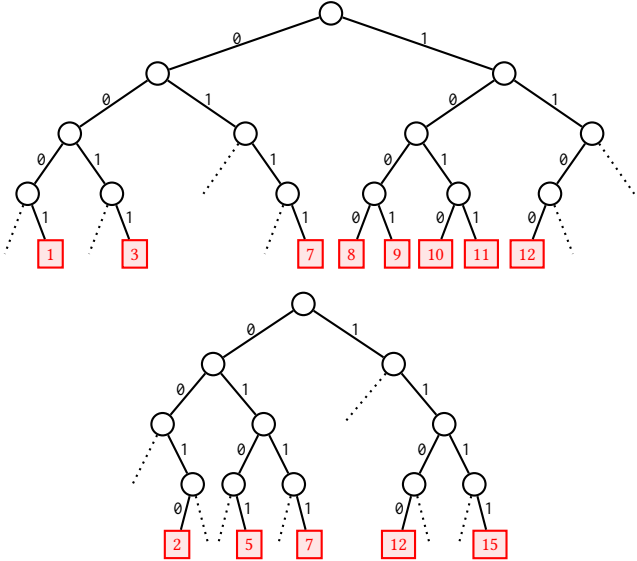
\begin{figure}[ht]
    \centering
\begin{tikzpicture}[
level distance=25mm,
level 1/.style={sibling distance=144mm},
level 2/.style={sibling distance=73mm},
level 3/.style={sibling distance=35mm},
level 4/.style={sibling distance=19mm},
thick, scale=0.32
]

\tikzstyle{c} = [draw, thick,shape=circle]
\tikzstyle{r} = [draw, thick,shape=rectangle,minimum width=1mm, red,fill=red!10!white]
\tikzstyle{tr} = [draw, thick,isosceles triangle, shape border rotate=90, anchor=north]

\node[c]{}[edge from parent]
    child {node[c] {}
        child {node[c]{}
           child { node[c]{}
              child { edge from parent [dotted];}
              child { node[r]{\footnotesize 1}
                  edge from parent coordinate (r0);
              }
              edge from parent coordinate (l0);
           }
           child { node[c]{}
               child { edge from parent [dotted];}
               child { node[r]{\footnotesize 3}
                   edge from parent coordinate (r1);
               }
               edge from parent coordinate (r2);
           }
           edge from parent coordinate (l1);
        }
        child { node[c]{}
            child { edge from parent [dotted];}
            child { node[c]{}
                child { edge from parent [dotted];}
                child { node[r]{\footnotesize 7}
                    edge from parent coordinate (r3);    
                }
                edge from parent coordinate (r4);
            }
            edge from parent coordinate (r5);
        }
        edge from parent coordinate (l2);
    }
    child {node[c]{}
        child {node[c]{}
            child{ node[c]{}
                child { node[r]{\footnotesize 8}
                    edge from parent coordinate (l3);
                }
                child { node[r]{\footnotesize 9}
                    edge from parent coordinate (r6);
                }
                edge from parent coordinate (l4);
            }
            child { node[c]{}
                child { node[r]{\footnotesize 10}
                    edge from parent coordinate (r77);
                }
                child { node[r]{\footnotesize 11}
                    edge from parent coordinate (r7);
                }
                edge from parent coordinate (r8);
            }
            edge from parent coordinate (l5);
        }
        child {node[c]{}
            child { node[c]{}
                child { node[r]{\footnotesize 12}
                    edge from parent coordinate (l6);
                }
                child { edge from parent [dotted];}
                edge from parent coordinate (l7);
            }
            child { edge from parent [dotted];}
            edge from parent coordinate (r9);
        }
        edge from parent coordinate (r10);
    };
\node at([xshift=-3mm,yshift=2mm]l0)  {\tt \footnotesize 0};
\node at([xshift=-3mm,yshift=2mm]l1)  {\tt \footnotesize 0};
\node at([xshift=-3mm,yshift=2mm]l2)  {\tt \footnotesize 0};
\node at([xshift=-3mm,yshift=2mm]l3)  {\tt \footnotesize 0};
\node at([xshift=-3mm,yshift=2mm]l4)  {\tt \footnotesize 0};
\node at([xshift=-3mm,yshift=2mm]l5)  {\tt \footnotesize 0};
\node at([xshift=-3mm,yshift=2mm]l6)  {\tt \footnotesize 0};
\node at([xshift=-3mm,yshift=2mm]l7)  {\tt \footnotesize 0};
\node at([xshift=-3mm,yshift=2mm]r77) {\tt \footnotesize 0};
\node at([xshift=3mm,yshift=2mm]r0) {\tt \footnotesize 1};
\node at([xshift=3mm,yshift=2mm]r1) {\tt \footnotesize 1};
\node at([xshift=3mm,yshift=2mm]r2) {\tt \footnotesize 1};
\node at([xshift=3mm,yshift=2mm]r3) {\tt \footnotesize 1};
\node at([xshift=3mm,yshift=2mm]r4) {\tt \footnotesize 1};
\node at([xshift=3mm,yshift=2mm]r5) {\tt \footnotesize 1};
\node at([xshift=3mm,yshift=2mm]r6) {\tt \footnotesize 1};
\node at([xshift=3mm,yshift=2mm]r7) {\tt \footnotesize 1};
\node at([xshift=3mm,yshift=2mm]r8) {\tt \footnotesize 1};
\node at([xshift=3mm,yshift=2mm]r9) {\tt \footnotesize 1};
\node at([xshift=3mm,yshift=2mm]r10) {\tt \footnotesize 1};
\end{tikzpicture}

\bigskip

\begin{tikzpicture}[
level distance=25mm,
level 1/.style={sibling distance=75mm},
level 2/.style={sibling distance=42mm},
level 3/.style={sibling distance=25mm},
level 4/.style={sibling distance=13mm},
 thick, scale=0.32
]

\tikzstyle{c} = [draw,  thick,shape=circle]
\tikzstyle{r} = [draw,  thick,shape=rectangle,minimum width=1mm, red,fill=red!10!white]
\tikzstyle{tr} = [draw, thick,isosceles triangle, shape border rotate=90, anchor=north]

\node[c]{}[edge from parent]
    child {node[c] {}
        child {node[c]{}
            child { edge from parent [dotted];}
            child { node[c]{}
                child { node[r]{\footnotesize 2}
                    edge from parent coordinate (l4);
                }
                child { edge from parent [dotted];}
                edge from parent coordinate (r8);
            }
            edge from parent coordinate (l5);
        }
        child { node[c]{}
            child { node[c]{}
                child { edge from parent [dotted]}
                child { node[r]{\footnotesize 5}
                    edge from parent coordinate (r0);
                }
                edge from parent coordinate (l0);
            }
            child { node[c]{}
                child { edge from parent [dotted];}
                child { node[r]{\footnotesize 7}
                    edge from parent coordinate (r1);    
                }
                edge from parent coordinate(r2)
            }
            edge from parent coordinate (r3);
        }
        edge from parent coordinate (l1);
    }
    child {node[c]{}
        child { edge from parent [dotted];}
        child {node[c]{}
            child { node[c]{}
                child { node[r]{\footnotesize 12}
                    edge from parent coordinate (l2);
                }
                child { edge from parent [dotted];}
                edge from parent coordinate (l3);
            }
            child { node[c]{}
                child { edge from parent [dotted];}
                child { node[r]{\footnotesize 15}
                    edge from parent coordinate (r4);
                }
                edge from parent coordinate (r5);
            }
            edge from parent coordinate (r6);
        }
        edge from parent coordinate (r7);
    };
\node at([xshift=-3mm,yshift=2mm]l0)  {\tt \footnotesize 0};
\node at([xshift=-3mm,yshift=2mm]l1)  {\tt \footnotesize 0};
\node at([xshift=-3mm,yshift=2mm]l2)  {\tt \footnotesize 0};
\node at([xshift=-3mm,yshift=2mm]l3)  {\tt \footnotesize 0};
\node at([xshift=-3mm,yshift=2mm]l4)  {\tt \footnotesize 0};
\node at([xshift=-3mm,yshift=2mm]l5)  {\tt \footnotesize 0};
\node at([xshift=3mm,yshift=2mm]r0) {\tt \footnotesize 1};
\node at([xshift=3mm,yshift=2mm]r1) {\tt \footnotesize 1};
\node at([xshift=3mm,yshift=2mm]r2) {\tt \footnotesize 1};
\node at([xshift=3mm,yshift=2mm]r3) {\tt \footnotesize 1};
\node at([xshift=3mm,yshift=2mm]r4) {\tt \footnotesize 1};
\node at([xshift=3mm,yshift=2mm]r5) {\tt \footnotesize 1};
\node at([xshift=3mm,yshift=2mm]r6) {\tt \footnotesize 1};
\node at([xshift=3mm,yshift=2mm]r7) {\tt \footnotesize 1};
\node at([xshift=3mm,yshift=2mm]r8) {\tt \footnotesize 1};
\end{tikzpicture}

    \caption{Binary tries \bintrie{S_1} and \bintrie{S_2} encoding sets $S_1 = \{1, 3, 7, 8, 9, 10, 11, 12\}$ and $S_2 = \{2, 5, 7, 12, 15\}$.}
    \label{fig:binary tries}
\end{figure}
From this encoding, the following compression measure can be derived \cite{GHSVtcs07}. Given two bit strings $x$ and $y$ of $\ell$ bits each, let $x \ominus y$ denote the bit string obtained after removing the longest common prefix among $x$ and $y$ from $x$. For instance, for $x = \texttt{0110100}$ and $y = \texttt{0111011}$, we have $x \ominus y = \texttt{0100}$. The prefix omission method by Klein and Shapira \cite{KSdcc02} represents a sorted set $S$ as a binary sequence $\mathcal{T} = \langle x_{1}; x_{2}\ominus x_{1}; \ldots; x_{n}\ominus x_{n-1}\rangle$. If we denote $|x_{i} \ominus x_{i-1}|$ the length of bit string $x_{i} \ominus x_{i-1}$, then the whole sequence uses

$$
\trie{S} = |x_{1}| + \sum_{i=2}^{n}{|x_{i} \ominus x_{i-1}|}.
$$
It turns out that \trie{S} is the number of edges in \bintrie{S} \cite{GHSVtcs07}. Notice that \trie{S} decreases as more trie paths are shared among set elements: consider two integers $x$ and $y$, the trie represents their longest common prefix just once (then saving space), and then represents both $x\ominus y$ and $y\ominus x$.
Extreme cases are as follows: (1) All set elements form a single run of consecutive elements, which maximizes the number of trie edges shared among set elements, hence minimizing the space usage; and (2) The $n$ elements are uniformly distributed within $[0{..}u)$ (i.e., the gap between successive elements is $g_i=u/n$), which minimizes the number of trie edges shared among elements, and hence maximizes space usage. 
Notice this is similar to the case that maximizes the $\gap{S}$ measure.

\begin{definition}
We say that a node $v$ in \bintrie{S} \emph{covers} all leaves that descend from it. In such a case, we call $v$ a \emph{cover node} of the corresponding leaves. 
\end{definition}

The following lemma summarizes several results that shall be important for our work:
\begin{lemma}[\cite{GNPtcs12}] \label{lemma:bintrie-properties}
For \bintrie{S}, the following results hold: 
\begin{enumerate}
    \item Any contiguous range of $L$ leaves in \bintrie{S} is covered by $\bigoh{\lg{L}}$ nodes.
    \item Any set of $r$ nodes in \bintrie{S} has $\bigoh{r\lg{\frac{\univSize}{r}}}$ ancestors.
    \item Any set of $r$ nodes in \bintrie{S} covering a contiguous range of leaves in the trie has $\bigoh{r + \lg{\univSize}}$ ancestors. \label{item:props-bintrie-3}
    \item Any set of $r$ nodes in \bintrie{S} covering subsets of $L$ contiguous leaves has $\bigoh{\lg{r} + r\lg{\frac{L}{r}}}$ ancestors. \label{item:props-bintrie-4}
\end{enumerate}
\end{lemma}

\begin{definition}
Given a set $S=\{x_1, \ldots, x_n\} \subseteq [0{..}u)$, let $S+a$, for $a \in [0{..}u)$, denote a shifted version of $S$: $S+a = \{(x_1+a)\bmod{u}, (x_2+a)\bmod{u},\ldots, (x_n+a)\bmod u \}$.
\end{definition}
The following result is relevant for our proposal:
\begin{lemma}[\cite{GHSVtcs07}] \label{lemma:trie}
Given a set $S\subseteq [0{..}u)$ of $n$ elements, it holds that:
\begin{enumerate}
\item $\trie{S} \le \min{\{2\gap{S},n\lg{(u/n)+2n}\}}$. \label{item:trie-1}
\item $\exists a \in [0{..}u)$, such that $\trie{S+a} \le \gap{S} + 2n -2$. \label{item:trie-2}
\item For any $a \in [0{..}u)$ chosen uniformly at random, $\trie{S+a} \le \gap{S} + 2n -2$ on average. \label{item:trie-3}
\end{enumerate}
\end{lemma}


\subsection{Set Intersection Algorithms} 

We review set intersection algorithms needed for our work.

\subsubsection{Barbay and Kenyon Algorithm}

Barbay and Kenyon intersection algorithm \cite{BKtalg08} is an \emph{adaptive approach} that works on the (more general) comparison model. An adaptive algorithm is one whose running time is a function not only of the instance size (as usual), but also of a difficulty measure of the instance. In this way, ``easy'' instances are solved faster than ``difficult'' ones, allowing for a more refined analysis than typical worst-case approaches. 
Given a query $\mathcal{Q} = \{i_1, \ldots, i_k\} \subseteq [1{..}N]$, the algorithm assumes the involved sets are represented using sorted arrays storing the $\lceil\lg{u}\rceil$-bit binary codes of the set elements. Algorithm \ref{alg:BarbayKenyon} shows the pseudocode.
    \begin{algorithm}[ht]
\DontPrintSemicolon
\KwResult{The set intersection $I = S_1\cap \cdots \cap S_k$}
    \Begin{
                $I \gets \emptyset$\\
                $x \gets S_1[1]$\\ 
                $i \gets 2$\\
                $occ \gets 1$ \\
                \While{$x \neq \infty$}{
                    $y \gets \successor(S_i, x)$\\
                    \If {$x = y$}{
                        $occ \gets occ + 1$
                    }
                    \If{$occ = k~\vee~x \neq y $}{
                        \If{$occ = k$}{
                            $I \gets I \cup \{x\}$
                        }
                        $x \gets y$ \\ 
                        $occ \gets 1$
                    } 
                    $i \gets (i+1)\bmod k$
                }
                \Return $I$
        }
        \caption{$\texttt{BK-Intersection}(\textrm{sets}~S_1, \ldots, S_k)$} 
        \label{alg:BarbayKenyon}
    \end{algorithm}
The algorithm keeps an element $x$ which is used to carry out searches in all sets. Initially, $x\gets S_1[1]$ (i.e., the first element of set $S_1$). At each step, the algorithm looks for the successor of $x$ in the next set (say, $S_2$), using doubling (or exponential) search \cite{BYipl76}. If $x \in S_2$, we repeat the process with $S_3$. Otherwise, $x \gets \successor(S_2, x)$, and continue with $S_3$. It is also important to keep a finger $f_i$ with every set $S_i$ in the query. The idea is that $f_i$ indicates the point where the previous search finished in $S_i$. So, the next search will start from $f_i$.


Any algorithm that computes $\mathcal{I(Q)}$ must show a certificate \cite{BKtalg08} or proof \cite{DLMsoda00} to show that the intersection is correct. That is, that any element in $\mathcal{I(Q)}$ belongs to the $k$ sets $S_{i_1}, \ldots, S_{i_k}$, and that no element in the intersection has been left out of the result. Barbay and Kenyon \cite{BKtalg08} introduced the notion of \emph{partition certificate} of a query $\mathcal{Q}$, which can be defined as follows.
\begin{definition}
\label{def:certificate}
Given a query $\mathcal{Q} = \{i_1, \ldots, i_k\} \subseteq [1{..}N]$,
its \emph{partition certificate} is a partition of the universe
$[0{..}\univSize)$ into a set of intervals $\mathcal{P} = \{I_1, I_2, \ldots, I_{|\mathcal{P}|}\}$, 
such that:
\begin{enumerate}
\item $\forall x \in \mathcal{I(Q)}, [x{..}x] \in \mathcal{P}$; 
\label{item:1-BK}
\item $\forall x \not \in \mathcal{I(Q)}, \exists I_j \in \mathcal{P}, x \in I_j~\wedge~\exists q \in \mathcal{Q}, S_q \cap I_j = \emptyset$.
\label{item:2-BK}
\end{enumerate}
\end{definition}
For a given query $\mathcal{Q}$, several valid partition certificates could be given. However, we are interested in the smallest partition certificate of $\mathcal{Q}$, as it takes the least time to be computed.
\begin{definition}
For a given query $\mathcal{Q} = \{i_1, \ldots, i_k\} \subseteq [1{..}N]$, let $\delta$ denote the size of the smallest partition certificate of $\mathcal{Q}$.
\end{definition}
This value $\delta$ is known as the \emph{alternation measure} \cite{BKtalg08}, and it (somehow) measures the difficulty of a given instance of the intersection problem. Notice $|\mathcal{I(Q)}| \le \delta$ holds.
Figure \ref{fig:partition-certificate} shows the smallest partition certificate (of size $\delta=8$) for sets $S_1$ and $S_2$ of our running example.
\begin{figure}
\begin{tabular}{ccc|c|cc|cc|c|cccc|c|ccc}
$S_{1}:$ & & 1 &   & 3 &  &   &  & 7 & 8 & 9 & 10 & 11 & 12 &  \\
$S_{2}:$ & &  & 2 &   &  & 5 &  & 7 &   &   &    &    & 12 &  & & 15 \\
\end{tabular}
\caption{Vertical lines show the smallest partition certificate
$\mathcal{P} = \{[0{..}1], [2{..}2], [3{..}4], [5{..}6], [7{..}7], [8{..}11], [12{..}12], [13{..}15]\}$ of size $\delta = 8$ of the universe $[0{..}16)$  
for the intersection of sets
$S_1 = \{1, 3, 7, 8, 9, 10, 11, 12\}$ and $S_2 = \{2, 5, 7, 12, 15\}$.
}
\label{fig:partition-certificate}
\end{figure}
Barbay and Kenyon \cite{BKsoda02,BKtalg08} proved a lower bound of $\Omega(\delta \sum_{i\in\mathcal{Q}}{\lg{\!(n_{i}/\delta)}})$ comparisons for the set intersection problem. They also proved that $\texttt{BK-Intersection}$ runs in $\bigoh{\delta\sum_{i\in\mathcal{Q}}{\lg{\!(n_i/\delta)}}}$ time, which is optimal.

\subsubsection{Trabb-Pardo Algorithm}
In his thesis, Trabb-Pardo \cite{TrabbPardo} studied integer-set representations and their corresponding intersection algorithms, such as the divide-and-conquer approach shown in Algorithm \ref{alg:set-intersection}.
\begin{algorithm}[ht]
\DontPrintSemicolon
\KwResult{The set intersection $S_1\cap \cdots \cap S_k$}
\Begin{
    \DontPrintSemicolon\Comment*[l]{Base cases}
      \For{$i\gets 1~\mathrm{to}~k$}{
          \If{$S_i = \varnothing$}{ 
              \Return{$\varnothing$}\;
          }
      }
    \eIf{$L = R$}{
        \Return{$\{L\}$} \DontPrintSemicolon\Comment*[r]{Universe of size 1, all sets are the same singleton}
    }{
      \DontPrintSemicolon\Comment*[l]{Divide}
      $M \gets (R+L)/2$\\
      \For{$i\gets 1~\mathrm{to}~k$}{
          $S_{i,l} \gets \{x \in S_i~|~ x\in [L{..}M)\}$ \label{line:divide-left}\\
          $S_{i,r} \gets \{x \in S_i~|~ x\in [M{..}R)\}$ \label{line:divide-right}\\
      }
      \DontPrintSemicolon\Comment*[l]{Conquer}
      $R_1 \gets \texttt{TP-Intersection}(S_{1,l}, \ldots, S_{k, l}, [L{..}M))$\\
      $R_2 \gets \texttt{TP-Intersection}(S_{1,r}, \ldots, S_{k, r}, [M{..}R))$\\
      \DontPrintSemicolon\Comment*[l]{Combine}
      \Return{$R_1 \cup R_2$} \DontPrintSemicolon\Comment*[r]{Disjoint set union}
    }
} 
 \caption{$\texttt{TP-Intersection}(\mathrm{sets}~S_1,\ldots,S_k;~\mathrm{universe}~[L{..}R))$}
\label{alg:set-intersection}
\end{algorithm}
Given a query $\mathcal{Q} = \{i_1, \ldots, i_k\}\subseteq [1{..}N]$, $\mathcal{I(Q)}$ is computed invoking $\texttt{TP-Intersection}(S_{i_1}, \ldots, S_{i_k}, [0{..}u))$. 
The main idea is to divide the universe into two halves, to then divide each set into two according to this universe division. This differs from, e.g., Baeza-Yates's algorithm \cite{BYcpm04,BYSspire05}, which divides according to the median of one of the sets.
The \emph{Divide} steps (lines \ref{line:divide-left} and \ref{line:divide-right}) can be implemented implicitly by binary searching for the successor of $(R+L)/2$, provided the sorted sets are represented using plain arrays. 
At the first level of recursion, notice that the most significant bit of every element in sets $S_{i,l}$ is \texttt{0} (i.e., belong to the left half of the universe), whereas for $S_{i,r}$ that bit is a \texttt{1} (i.e., they belong to the right half of the universe). At each node of the recursion tree, the current universe is divided into 2 halves, and then we recurse on the sets divided accordingly. 
The universe partition carried out by Trabb-Pardo's algorithm avoids several of the problems faced by Baeza-Yates' algorithm and allows us for a more efficient implementation with time guarantees, as we shall see.  

Sets are known in advance, so the Divide step of Algorithm \ref{alg:set-intersection} can be implemented efficiently by precomputing the set divisions carried out recursively. This is because set divisions are carried out according to the universe, which is query independent. 
Algorithms that divide according to the median element \cite{BYcpm04,BYSspire05}, on the other hand, have query dependent divisions, hence hard to be precomputed.  Trabb-Pardo proposed to store the precomputed divisions using a binary trie representations of the sets (see Section \ref{sec:trie-measure}). We store \bintrie{S_i} for each $S_i\in \mathcal{S}$. Notice how \bintrie{S_i} mimics the way set $S_i$ is recursively divided by Algorithm \ref{alg:set-intersection}. The left child of the root represents all elements whose most significant bit is \texttt{0}, that is, the elements in $S_{i,l}$ (see line \ref{line:divide-left}) of Algorithm \ref{alg:set-intersection}; similarly for $S_{i,r}$, which contains all elements in $S_i$ whose most-significant bit is  \texttt{1}.


Algorithm \ref{alg:set-intersection} is implemented on binary tries as follows. 
In particular, $\texttt{TP-Intersection}$ induces a DFS traversal in synchronization on all tries involved in the query, following the same path in all of them and stopping (and backtracking if needed) as soon as we reach a dead end in one of the tries (which correspond to dotted lines in Figure \ref{fig:binary tries}), or we reach a leaf node in all the tries (in whose case we have found an element belonging to the intersection). In this way, (1) we stop as soon as we detect a universe interval that does not have any element in the intersection, and (2) we find the relevant elements when we arrive at a leaf node.  

\subsection{Compressed Set Representations}

We review next some 
compressed set representations.

\subsubsection{Integer Compression}
A rather natural way to represent a set $S$ is to use \emph{universal codes} \cite{MG} for integers to represent either the sequence of gaps $\mathcal{G}$, the run lengths $\mathcal{Z}$ and $\mathcal{O}$, or the POM sequence $\mathcal{T}$.
Elias $\gamma$ and $\delta$ \cite{Elias1975} are among the most know universal codes, as well as Fibonacci codes \cite{FKdam96}. 
In particular, Elias $\delta$ encodes an integer $x$ using $\lfloor \lg{x}\rfloor + 2 \lfloor\lg{(\lfloor\lg{x} \rfloor+1)} \rfloor$ bits. 
Hence, sequence $\mathcal{G}$ can be encoded using $\sum_{i=1}^{n}{\lfloor \lg{x}\rfloor + 2 \lfloor\lg{(\lfloor\lg{x} \rfloor+1)} \rfloor} = \gap{S}(1+o(1))$ bits. Similarly, one can achieve space close to $\rle{S}$ or $\trie{S}$ by using Elias $\delta$ of sequences $\mathcal{Z}$ and $\mathcal{O}$ or $\mathcal{T}$, respectively.  

As we encode the gap sequence of the set, direct access to the set elements in no longer possible. To obtain element $x_i \in S$ we must decompress $g_1, \ldots, g_i$ to compute $x_i = \sum_{j=1}^{i}{g_j}$, which takes $\bigoh{i}$ time. To reduce this cost, the sequence of codes is divided into $\lceil n/B\rceil$ blocks of $B$ elements each. We store an array $A[1{..}\lceil n/B \rceil]$ such that $A[i] = \sum_{j=1}^{i\lceil n/B \rceil}{g_j}$, adding $\lceil n/B \rceil\lg{u}$ extra bits. Now, decoding $x_i$ takes $\bigoh{B}$ time, as only the block containing $x_i$ must be accessed. This representation is typical in inverted index compression in IR \cite{BCC2010,MG}. Intersections can be computed using \texttt{BK-Intersection}, using array $A$ to skip unwanted blocks. 
The running time is increased by up to $\bigoh{\delta B}$. Besides, decoding a single integer encoded with Elias $\delta$ or Fibonacci is slow in practice ---15--30 nsecs per decoded integer is usual \cite{BCC2010}. Hence, alternative approaches are used in practice, such as VByte \cite{WZcj99}, Simple9 \cite{AMir05} (and optimized variants like Simple16 \cite{ZLSwww08} and Simple18 \cite{AOGSipm18}), and PForDelta \cite{ZHNBicde06} (and optimized variants like OptPFD \cite{YDSwww09}). These have efficient decoding time ---e.g., less than 1 nsec/int on average is typical \cite{BCC2010}---, yet their space usage is not guaranteed to achieve any compression measure, although they yield efficient space usage in practice.

Alternatively, Elias-Fano \cite{Ejacm74,fano1971number,OSalenex07,Vwsdm13} compression represents set elements directly, rather than their gaps, using  $n\lg{(u/n)} + 2n + o(n)$ bits of space \cite{OSalenex07} (which is worst-case-optimal \footnote{Actually, this is almost worst-case optimal, as this is $\Theta(n)$ bits above the worst-case lower bound defined before in this paper.}) and supporting efficient access to set elements. In general, this representation uses more space than Elias $\gamma$ and $\delta$ over the gap sequences (recall Lemma \ref{lemma:gap-vs-wc-lower-bound}), however, e.g., Okanohara and Sadakane \cite{OSalenex07} support access to a set element $x_i$ in $\bigoh{1}$ time. Hence, algorithm $\texttt{BK-Intersection}$ can be implemented in $\bigoh{\delta\sum_{i=1}^{k}{(\lg{\frac{u}{n_i}} + \lg{\frac{n_i}{\delta}})}} = \bigoh{\delta\sum_{i=1}^{k}{\lg{\frac{u}{\delta}}}}$. Recent variants of Elias-Fano are Partitioned Elias-Fano (PEF) by Ottaviano and Venturini \cite{OVsigir14} and Clustered Elias-Fano by Pibiri and Venturini \cite{PVtois17}. These are highly competitive approaches in practice, taking advantage of the non-uniform distribution of set elements along the universe.

The approach of Ding and Konig is also worth to be mentioned \cite{DKpvldb11}. They introduce a linear-space data structure that leverages the word RAM model of computation to compute intersections in time $\bigoh{n/\sqrt{w} + k|\mathcal{I(Q)}|}$, where $n = \sum_{i\in\mathcal{Q}}{n_i}$. According to their experiments, this seems to be a practical approach. Another practical studies of set intersection algorithms are by Tsirogiannis et al.~\cite{TGKpvldb09} and Kim et al.~\cite{KLHEpvldb18}.

\subsubsection{Compressed $\rank/\select$ Data Structures}

Compressed data structures supporting $\rank$ and $\select$ operations can be used to support successor (i.e., $\successor (S, x) \equiv \select(S, \rank(S, x-1) + 1)$) and $S[i] \equiv \select(S, i)$, needed to implement $\texttt{BK-Intersection}$. We survey next the most efficient approaches on these lines.

Operations $\rank$ and $\select$ can be supported in $\bigoh{c}$ time using a combination of Raman, Raman, and Rao succinct representation \cite{RRR07} and \patrascu's succincter \cite{Patrascu08}, using $\lowerB + \bigoh{u/\lg^c{(u/c)}} + \bigoh{u^{3/4}\textrm{poly}\log{(u)}}$ bits, for any $c > 0$. Algorithm $\texttt{BK-Intersection}$ on this representation takes $\bigoh{ck\delta}$ time. However, the big-oh term in the space usage depends on the universe size, which would introduce an excessive space usage if the universe is big.
The approaches of M{\"{a}}kinen and Navarro \cite{MNtcs07} and Sadakane and Grossi \cite{SGsoda06} can be used to reduce the term $\lowerB$ to $\gap{S}$, while retaining $\bigoh{c}$ time $\rank$ and $\select$. Arroyuelo and Raman \cite{ARalg22} were able to squeeze the space usage further, to achieve better entropy bounds. They support $\rank$ and $\select$ in $\bigoh{1}$ time while using $n\hzgap{S} + \bigoh{u(\lg\lg{u})^2/\lg{u}}$ bits, or alternatively $r\hzrun{S} + \bigoh{u(\lg\lg{u})^2/\lg{u}}$ bits, where $\hzgap{S}$ and $\hzrun{S}$ are the zero-order entropies of the gap and run-length distributions, respectively, and $r$ is the number of runs of successive elements in $S$. Algorithm $\texttt{BK-Intersection}$ takes $\bigoh{k\delta}$ time on this representation. However, again the $o(u)$ term still dominates for big universes.
    
To avoid the $o(u)$-bit dependence, Gupta et al.~\cite{GHSVtcs07} data structure uses $\gap{S} (1+o(1))$ bits,  supporting $\rank$ in $\bigoh{\PT(u,n,a) + \lg{\lg{n}}}$ time and $\select$ in $\bigoh{\lg{\lg{n}}}$ time, where $\PT(u,n,a)$ is \patrascu and Thorup \cite{PT} optimal bound for a universe of size $u$, $n$ elements, and $a = \lg{(\lg{u}/\lg^2{n})}$ \cite{ARalg22}. 
Hence, \texttt{BK-Intersection} takes  $\bigoh{\delta\sum_{i=1}^{k}{(\PT(u, n_i, a_i) + \lg{\lg{n_i}})}}$ time on this representation, where $a_i = \lg{(\lg{u}/\lg^2{n_i})}$. Alternatively, one can use $\rle{S} + o(\rle{S})$ bits of space, with $\bigoh{\delta\sum_{i=1}^{k}{(\PT(\univSize, r_i, a_i) + \linebreak[4] (\lg{\lg{r_i}})^2)}}$ intersection time \cite{ARalg22}, for $a_i=\lg{(\lg{(\univSize-n_i)}/\lg^2{r_i})}$, where $r_i$ is the number of runs of successive elements in set $S_i$, $i \in \mathcal{Q}$.
Finally, the data structures by Arroyuelo and Raman \cite{ARalg22} use  $(1+1/t) n \hzgap{S} + \bigoh{n / \lg u}$ bits or $(1+1/t) r \hzrun{S} + \bigoh{n / \lg u}$  bits of space, for $t > 0$. \texttt{BK-Intersection} on them takes $\bigoh{\delta\sum_{i=1}^{k}{(t + \PT(u,n_i,\bigoh{1}))}}$ time. Although they avoid the dependence on $u$, their space still depends on $n$. This would dominate when the $n$ elements are strongly grouped into runs, blowing up the space usage.

Finally, it is worth noting that wavelet trees's \cite{GGVsoda03} intersection algorithm by Gagie et al.~\cite{GNPtcs12} actually carries out Trabb-Pardo's approach, adapted to work on this particular data structure. They show that their algorithm works in $\bigoh{k\delta\lg(\!u/\delta)}$ time, yet their space usage to represent a set would be $n\lg{u}$ bits, for a set of $n$ elements.

%% file: 03-compressed-tries.tex
\section{Trie Certificates}

Next, we analyze the running time of $\texttt{TP-Intersection}$ when implemented using binary tries. Trabb-Pardo carried out just an average-case analysis of his algorithm \cite{TrabbPardo}. We introduce the concept of \emph{trie certificate} to carry out an adaptive-case analysis. 

We show that the recursion tree (which we denote $\certificate{\mathcal{Q}}$) of the synchronized DFS traversal carried out by $\texttt{TP-Intersection}$ can act as a certificate (or proof) for the intersection. Figure \ref{fig:intersection-certificate} shows the binary trie \certificate{\mathcal{Q}} for $S_1 \cap S_2$, for sets $S_1$ and $S_2$ from Figure \ref{fig:binary tries}.
\begin{figure}[ht]
    \centering
\begin{tikzpicture}[
level distance=25mm,
level 1/.style={sibling distance=144mm},
level 2/.style={sibling distance=73mm},
level 3/.style={sibling distance=35mm},
level 4/.style={sibling distance=19mm},
thick, scale=0.32
]

\tikzstyle{c} = [draw,  thick,shape=circle]
\tikzstyle{r} = [draw,  thick,shape=rectangle,minimum width=1mm, red,fill=red!10!white]
\tikzstyle{tr} = [draw, thick,isosceles triangle, shape border rotate=90, anchor=north]

\node[c]{}[edge from parent]
    child {node[c] {}
        child {node[c]{}
            child { edge from parent [dotted];}
            child { node[c]{}
                child {edge from parent [dotted];}
                child {edge from parent [dotted];}
                edge from parent coordinate (r0);
            }
            edge from parent coordinate (l0);
        }
        child { node[c]{}
            child { edge from parent [dotted];}
            child { node[c]{}
                child { edge from parent [dotted];}
                child { node[r]{\footnotesize 7}
                    edge from parent coordinate (r1);    
                }
                edge from parent coordinate(r2)
            }
            edge from parent coordinate (r3);
        }
        edge from parent coordinate (l1);
    }
    child {node[c]{}
        child { edge from parent [dotted];}
        child {node[c]{}
            child { node[c]{}
                child { node[r]{\footnotesize 12}
                    edge from parent coordinate (l2);
                }
                child { edge from parent [dotted];}
                edge from parent coordinate (l3);
            }
            child { edge from parent [dotted];}
            edge from parent coordinate (r4);
        }
        edge from parent coordinate (r5);
    };
\node at([xshift=-3mm,yshift=2mm]l0)  {\tt \footnotesize 0};
\node at([xshift=-3mm,yshift=2mm]l1)  {\tt \footnotesize 0};
\node at([xshift=-3mm,yshift=2mm]l2)  {\tt \footnotesize 0};
\node at([xshift=-3mm,yshift=2mm]l3)  {\tt \footnotesize 0};
\node at([xshift=3mm,yshift=2mm]r0) {\tt \footnotesize 1};
\node at([xshift=3mm,yshift=2mm]r1) {\tt \footnotesize 1};
\node at([xshift=3mm,yshift=2mm]r2) {\tt \footnotesize 1};
\node at([xshift=3mm,yshift=2mm]r3) {\tt \footnotesize 1};
\node at([xshift=3mm,yshift=2mm]r4) {\tt \footnotesize 1};
\node at([xshift=3mm,yshift=2mm]r5) {\tt \footnotesize 1};
\end{tikzpicture}

    \caption{Trie certificate for the intersection $\{1, 3, 7, 8, 9, 10, 11, 12\} \cap  \{2, 5, 7, 12, 15\}$. This trie shows the nodes that must be checked to determine that the result is, in this case, $\{7,12\}$.}
    \label{fig:intersection-certificate}
\end{figure}
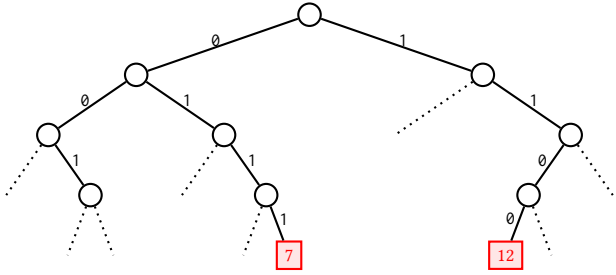
Notice that:
(1) every internal node in $\certificate{\mathcal{Q}}$ is a comparison that cannot discard elements, so we must keep going down; (2) every external node at depth $d \le \ell$ (which descend from an internal node using a dashed edge in Figure \ref{fig:intersection-certificate}) is a comparison that allows us discard the corresponding universe interval of $2^{\ell-d}$ elements. For instance, the leftmost dashed edge in Figure \ref{fig:intersection-certificate} corresponds to interval $[0{..}1]$, whereas the rightmost corresponds to $[14{..}15]$. More formally, assume a binary string $b$ of length $d-1 \le \ell$ such that $b$ belongs to all $\bintrie{S_{i_1}}, \ldots, \bintrie{S_{i_k}}$, and  $b\cdot\bit{0}$ (or, alternatively, $b\cdot\bit{1}$) belongs to all $\bintrie{S_i}$, $i\in\mathcal{Q}$, except for at least one of the  tries, say $\bintrie{S_j}$. That means that the universe interval $[b\cdot\bit{0}^{\ell-d}{..}b\cdot\bit{1}^{\ell-d}]$ (with endpoints represented in binary) has no elements in the intersection, and hence we can safely stop at path $b\cdot\bit{0}$; and (3) every external node at depth $\ell$ is an element belonging to the intersection.
In this way, the recursion tree is a certificate for the intersection. 

Overall, notice that $\certificate{\mathcal{Q}}$ covers the universe $[0{..}u)$ with intervals, indicating which universe elements belong to $\mathcal{I(Q)}$, and which ones do not. This is similar to the partition certificate delivered by Barbay and Kenyon's algorithm. For instance, the trie certificate of Figure \ref{fig:intersection-certificate} partitions the universe $[0{..}16)$ into $\{[0{..}1], [2{..}2], [3{..}3],\linebreak[4] [4{..}5], [6{..}6], [7{..}7], [8{..}11], [12{..}12], [13{..}13], [14{..}15]\}$.

\begin{definition}
\label{def:trie-certificate}
Given a query $\mathcal{Q} = \{i_1, \ldots, i_k\} \subseteq [1{..}N]$,
its \emph{trie partition certificate} is a partition of the universe
$[0{..}\univSize)$ into a set of intervals $\mathcal{P} = \{I_1, I_2, \ldots, I_{|\mathcal{P}|}\}$, 
induced by the trie $\certificate{Q}$.
\end{definition}

Let $|\certificate{\mathcal{Q}}|$ denote the number of nodes of $\certificate{\mathcal{Q}}$. Since $\certificate{\mathcal{Q}}$ is the recursion tree of algorithm \texttt{TP-intersection}, the time spent by the algorithm is $\bigoh{k|\certificate{\mathcal{Q}}|}$. 
To bound $k\certificate{\mathcal{Q}}$ in the worst-case , notice that at most we must traverse all tries $\bintrie{S_i}$, $i\in\mathcal{Q}$, so we have: 
\begin{align*} 
k|\certificate{\mathcal{Q}}| & \le \sum_{i\in \mathcal{Q}}{\trie{S_i}}
        \le \sum_{i \in \mathcal{Q}}{n_i\lg{\frac{u}{n_i}}}.
\end{align*}

Next we bound $k|\certificate{\mathcal{Q}}|$ in an adaptive way, to show that Algorithm \ref{alg:set-intersection} is actually an adaptive approach.
\begin{theorem} \label{theor:TP-adaptive}
Given a query instance $\mathcal{Q} = \{i_1,\ldots, i_k\}\subset [1{..}N]$ with alternation measure $\delta$, $\texttt{TP-Intersection}$ computes $\mathcal{I(Q)}$ in time $\bigoh{k\delta\lg{\!(u/\delta)}}$.
\end{theorem}
\begin{proof}
Consider the smallest partition certificate $\mathcal{P}$ of query $\mathcal{Q}$, which  partitions the universe into $\delta$ intervals. Let $L_1, \ldots, L_{\delta}$ be the size of each of the intervals $I_1, \ldots, I_{\delta}$ in $\mathcal{P}$. Let us consider the worst-case trie $\certificate{\mathcal{Q}}$ we could have. This can be obtained by covering the $\delta$ intervals with trie nodes. The number of cover nodes equals the size of the worst-case trie we could have, and hence it is a bound on the time spent computing the intersection.
Lets consider an interval $I_j$ formed by elements not in $\mathcal{I(Q)}$. Notice that when traversing the tries $\bintrie{S_i}$ in coordination, $i\in \mathcal{Q}$, as long as one gets into one of the cover nodes of $I_j$, the algorithm stops at that node because it does not belong to at least one of the tries. According to Lemma \ref{lemma:bintrie-properties} (1), a contiguous range of $L$ leaves can be covered with up to $\bigoh{\lg{L}}$ nodes. Thus, in the worst-case, $\certificate{\mathcal{Q}}$ has $\bigoh{\sum_{i=1}^{\delta}{\lg{L_i}}}$ external nodes that overall cover $[0{..}u)$. Now, according to Lemma \ref{lemma:bintrie-properties} (\ref{item:props-bintrie-3}), 
these external nodes have $\bigoh{\sum_{i=1}^{\delta}{\lg{L_i}} + \lg{u}}$ ancestors, so overall $\certificate{\mathcal{Q}}$ has $\bigoh{\sum_{i=1}^{\delta}{\lg{L_i}} + \lg{u}}$ nodes. The sum is maximized when $L_i = u/\delta$, for all $1\le i \le \delta$, hence $\certificate{\mathcal{Q}}$ has $\bigoh{\delta\lg{(u/\delta)}}$ nodes. As for each node in $\certificate{\mathcal{Q}}$ we must pay time $\bigoh{k}$, and the result follows. 
\end{proof}

\section{Compressed Intersectable Sets} \label{sec:compressed-trie}

We devise next a space-efficient representation of $\bintrie{S}$, for a set $S = \{x_1, \ldots, x_n\} \subseteq [0{..}\univSize)$ of $n$ elements such that $0 \le x_1 < \cdots < x_n < \univSize$. This representation will also allow for efficient intersections, supporting Trabb-Pardo's \cite{TrabbPardo} algorithm.

\subsection{A Space-Efficient \texorpdfstring{$\bintrie{S}$}{bintrie(S)}}

We represent $\bintrie{S}$ level-wise. Let $B_1[1{..}2l_1], \ldots, B_{\ell}[1{..}2l_{\ell}]$ be bit vectors such that $B_i$ stores the $l_i$ nodes at level $i$ of $\bintrie{S}$ ($1\le i\le \ell$), from left to right. Each node is encoded using 2 bits, indicating the presence (using bit \bit{1}) or absence (bit \bit{0}) of the left and right children. In this way, 
the feasible node codes are $\bit{01}$, $\bit{10}$, and $\bit{11}$. Notice that $\bit{00}$ is not a valid node code: that would mean that the entire subtree of a node is empty, contradicting the fact that the path to which the node belongs has at least one element (as the trie only expands the paths that contain elements). The node codes of all nodes at level $i\ge 1$ in the trie are concatenated from left to right to form $B_i$.
The $j$-th node at level $i$ (from left to right) is stored at positions $2(j-1)+1$ and $2(j-1)+2$.

Let $p$ be the position in $B_i$ corresponding to a node $v$ at level $i$ of $\bintrie{S}$. 
As the nodes are stored level-wise, the number of \bit{1}s before position $p$ in $B_i$ equals the number of nodes in $B_{i+1}$ that are before the child(ren) of node $v$. So, $2\cdot B_{i}.\rank_{\bit{1}}(p-1)$ yields the position of $B_{i+i}$ where the child(ren) of node $v$ are. 
Figure \ref{fig:trie-BV} illustrates our representation.

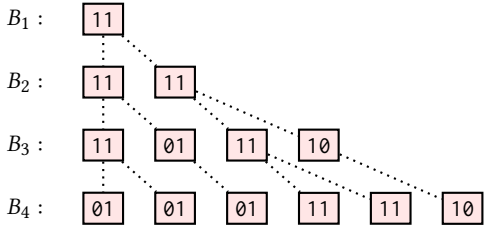
\begin{figure}[ht]
    \centering
\begin{minipage}[m]{0.49\textwidth}
\centering
\begin{tikzpicture}[
    scale=0.35, node distance=0.4cm
]

\begin{scope}[every node/.style={draw, thick,shape=rectangle,minimum width=1mm, black,fill=red!10!white}]
\node[] (A) {\texttt{11}};
\node[left = of A, draw=white, fill = white] {$B_1:$};

\node[below = of A] (B) {\texttt{11}};
\node[right = of B] (C) {\texttt{11}};
\node[left = of B, draw=white, fill = white] {$B_2:$};

\node[below = of B] (D) {\texttt{11}};
\node[right = of D] (E) {\texttt{01}};
\node[right = of E] (F) {\texttt{11}};
\node[right = of F] (G) {\texttt{10}};
\node[left = of D, draw=white, fill = white] {$B_3:$};

\node[below = of D] (H) {\texttt{01}};
\node[right = of H] (I) {\texttt{01}};
\node[right = of I] (J) {\texttt{01}};
\node[right = of J] (K) {\texttt{11}};
\node[right = of K] (L) {\texttt{11}};
\node[right = of L] (M) {\texttt{10}};
\node[left = of H, draw=white, fill = white] {$B_4:$};

\end{scope}

\begin{scope}
\draw[dotted, thick] (A)  to node {} (B);
\draw[dotted, thick] (A)  to node {} (C);

\draw[dotted, thick] (B)  to node {} (D);
\draw[dotted, thick] (B)  to node {} (E);

\draw[dotted, thick] (C)  to node {} (F);
\draw[dotted, thick] (C)  to node {} (G);

\draw[dotted, thick] (D)  to node {} (H);
\draw[dotted, thick] (D)  to node {} (I);

\draw[dotted, thick] (E)  to node {} (J);

\draw[dotted, thick] (F)  to node {} (K);
\draw[dotted, thick] (F)  to node {} (L);

\draw[dotted, thick] (G)  to node {} (M);
\end{scope}

\end{tikzpicture}
\end{minipage}
\caption{Level-wise bit vector representation of \bintrie{S} for $S = \{1, 3, 7, 8, 9, 10, 11, 12\}$. Dotted lines are just for clarity, as they are 
computed using operation $\rank_{\bit{1}}$ on the bit vectors.}
\label{fig:trie-BV}
\end{figure}

Notice that the total number of $\bit{1}$s in the bit vectors of our representation equals the number of edges in the trie. That is, there are $\trie{S}$ \bit{1}s. Besides, there are $\trie{S}+1$ nodes: $n$ of them are external, so $\trie{S}-n+1$ are internal. 
In our representation we only need to represent the internal trie nodes. As we encode each node using $2$ bits, the total space usage for $B_1, \ldots, B_\ell$ is $2(\trie{S}-n+1)$ bits. On top of them we use Clark's data structure \cite{ThesisClark} to support $\rank$ and $\select$ in $\bigoh{1}$ time, adding $o(\trie{S})$ extra bits.

\subsection{Supporting \texorpdfstring{$\rank$}{rank} and \texorpdfstring{$\select$}{select}} 

\subsubsection{Operation $\rank(S, x)$}

The main idea is to transform $\rank(S, x)$ into $B_{\ell}.\rank_{\bit{1}}(p)$, for a given position $p$ that we  compute as follows. Basically, we use the binary code of $x$ to go down from the root of $\bintrie{S}$. If $x \in S$, we will eventually reach the corresponding trie leaf at position $p$ in $B_\ell$. Then, $\rank(S, x) \equiv B_\ell.\rank_{\bit{1}}(p)$, as said before. The case $x\not \in S$ is more challenging. During the traversal, we keep the last trie node $v_{\mathrm{last}}$ along this path such that: (1) $v_{\mathrm{last}}$ is a branching node, and (2) the search path continues within the right child of $v_{\mathrm{last}}$, and (2) the left child of $v_{\mathrm{last}}$ exists (i.e., $v_{\mathrm{last}}$ is encoded $\bit{11}$).  As $x\not\in S$, we will eventually reach an internal trie node at level $l\ge 1$ that has no child corresponding to the $l$-th most-significant bit in the binary code of $x$. At this point, we must look for the leaf corresponding to $y = \predecessor(S, x)$. Notice $y$ is the largest value stored within the left subtree of node $v_{\mathrm{last}}$ (i.e., the rightmost leaf in that subtree). Thus, we start from the left child of $v_{\mathrm{last}}$, descending always to the right child. Upon reaching a node with no right child, we must go down to the left (which exists for sure), and then continue the process going to the right child again. Eventually, we will reach the corresponding trie leaf at position $p$ in $B_\ell$, and will be able to compute $\rank(S, x)$ as before.


\subsubsection{Operation $\select(S, j)$}

The $j$-th \bit{1} in $B_\ell$ corresponds to the $j$-th element in $S$. So, to compute $\select(S, j)$ we must start from $i \gets B_{\ell}.\select_{\bit{1}}(j)$, and then go up to the parent to compute the binary representation of $x_j$. Given a position $i$ in $B_{l}$, its parent can be computed as $B_{l-1}.\select_{\bit{1}}(\lceil i/2 \rceil)$. When going to the parent of the current node $v$ in the trie, we only need to know if $v$ is the right or left child of its parent, as this allows us to know each bit in the binary representation of $x_j$. In our representation, the left child of a node always correspond to an odd position within the corresponding bit vectors, whereas the right child correspond to an even position.




\subsection{Supporting Intersections}

We assume that for each $S_i\in\mathcal{S}$, $\bintrie{S_i}$ has been represented using our space-efficient trie representation from the previous subsection (that we denote $T_i$). Given a query $\mathcal{Q} = \{i_1, \ldots, i_k\} \subset [1{..}N]$, we traverse the tries $\bintrie{S_{i_1}}, \ldots , \bintrie{S_{i_k}}$ using the same recursive DFS traversal as in algorithm $\texttt{TP-Intersection}$. Besides the query itself, our algorithm receives:
    (1) an integer value, $level$, indicating the current level in the recursion, and 
    (2) the integer values $r_1, \ldots, r_k$, indicating the current nodes in each trie. These are the positions of the current nodes in each trie within their corresponding bit vectors $B_{level}$.
Algorithm \ref{alg:trie-intersection} shows the pseudo-code of our algorithm, which we call $\texttt{AC-Intersection}$, where AC stands for Adaptive and Compressed. $\texttt{AC-Intersection}$ computes our compact representation for $\bintrie{\mathcal{I(Q)}}$, which we denote $T_I$. The algorithm uses a variable $s$, initialized with $\bit{11}$, which will store the bitwise-and of all current node codes (computed on line \ref{line:bit-wise-and}). In this way, $s =\bit{00}$ indicates that recursion must be stopped at this node, $s = \bit{10}$ indicates to go down only to the left child, $s = \bit{01}$ to go down just to the right child, and $s = \bit{11}$ indicates to go down on both children. 

In lines \ref{line:if-left-child}--\ref{line:left-recursion} we carry out the needed computation to go down to the left child. We first determine whether we must go down to the left or not. In the former case, we compute the positions of the left-subtrie roots using $\rank$ operation. Then, on line \ref{line:left-recursion} we recursively go down to the left. The result of that recursion in stored in variable $lChild$, indicating with a $\bit{1}$ that the left recursion yielded a non-empty intersection, $\bit{0}$ otherwise. A similar procedure is carried out for the right children in lines \ref{line:if-right-child}--\ref{line:right-recursion}. Line \ref{line:if-left-rank} determines whether we have already computed the $\rank$s corresponding to the left child. If that is not the case, we compute them in line \ref{line:rank-right}; otherwise we avoid them. In this way we ensure the computation of only one $\rank$ operation per traversed node in the tries. Although $\rank$ can be computed in constant time, this is important in practice. Just as for the left child, we store the result of the right-child recursion in variable $rChild$ in line \ref{line:right-recursion}. Finally, in line \ref{line:if-empty-intersection} we determine whether the left and right recursions yielded an empty intersection or not. If both $lChild = 0$ and $rChild = 0$, the intersection was empty on both children. In such a case we return $\bit{0}$ indicating this fact. Otherwise, we append the value of $lChild$ and $rChild$ to $T_I.B_{level}$, as that is just the encoding of the corresponding node in the output trie $T_I$. Note how we actually generate the output trie $T_I$ in postorder, after we visited both children of the current nodes. This way, we write the output in time proportional to its size, saving time in practice.

\begin{algorithm}[ht]
\KwResult{The binary trie $T_I$ representing  $\mathcal{I(Q)}=\cap_{i\in\mathcal{Q}}{S_i}$}
\Begin{
    $s \gets \bit{11}$ \DontPrintSemicolon\Comment*[]{binary encoding}
    \For{$i \in \mathcal{Q}$}{
       $s \gets s~\&~ (T_{i}.B_{level}\left[r_i\right] \cdot T_{i}.B_{level}\left[r_i+1\right])$ \label{line:bit-wise-and}
    }
    \If{$level = \ell$}{
        $\mathtt{append}~s~\mathrm{to}~T_I.{B_\ell}$\\
        \Return{$\bit{1}$} \label{line:return-level-l}
    }
    $lChild\gets \bit{0};~rChild \gets \bit{0}$\\
    \DontPrintSemicolon\Comment*[]{Go down to the left in the tries}
    \If {$s~\mathtt{is}~\bit{10}~\mathtt{or}~\bit{11}$ \label{line:if-left-child}}{
        $lRoots\gets~\varnothing$\\ 
        \For{$i \in \mathcal{Q}$}{
            $lRoots\gets~lRoots~\cup~\{2\times T_i.B_{level}.\rank_{\bit{1}}(r_i)\}$
        }
        $lChild \gets \texttt{AC-Intersection}(\mathcal{Q}, lRoots, level + 1)$ \label{line:left-recursion}
    }
    \DontPrintSemicolon\Comment*[]{Go down to the right in the tries}
    \If {$s~\mathtt{is}~\bit{01}~\mathtt{or}~\bit{11}$ \label{line:if-right-child}}{
        $rRoots\gets~\varnothing$\\ 
        \For{$i \in \mathcal{Q}$}{
            \eIf{$s = \bit{01}$ \label{line:if-left-rank}}{
                $rRoots\gets~rRoots~\cup~\{2\times T_i.B_{level}.\rank_{\bit{1}}(r_i) + 1\}$ \label{line:rank-right}
            }{
                $rRoots\gets~rRoots~\cup~\{lRoots_i + 1\}$
            }   
        }
        $rChild \gets \texttt{AC-Intersection}(\mathcal{Q}, rRoots, level + 1)$ \label{line:right-recursion}
    }
    \DontPrintSemicolon\Comment*[]{Output written in postorder}
    \eIf{$lChild\not = \bit{0}~\mathbf{or}~rChild\not = \bit{0}$ \label{line:if-empty-intersection}}{
        $\mathtt{append}~lChild\cdot rChild~\mathtt{to}~T_I.{B_{level}}$\\
        \Return{$\bit{1}$}
    }{
        \Return{$\bit{0}$}
    }
} 
\caption{$\texttt{AC-Intersection}(\mathrm{query}~\mathcal{Q};~\mathrm{roots}~r_1,\ldots,r_k;~level)$}
\label{alg:trie-intersection}
\end{algorithm}

Besides computing $\mathcal{I(Q)}$, a distinctive feature of our algorithm is that it also allows one to obtain for free the sequence $\langle \rank(S_{i_1}, x), \linebreak[4] \ldots, \rank(S_{i_k}, x) \rangle$, for all $x \in \mathcal{I(Q)}$ . The idea is that every time we reach level $\ell$ of the tries, we compute $\langle T_1.B_{\ell}.\rank_{\bit{1}}(r_1), \linebreak[3] \ldots, \linebreak[3]  T_k.B_{\ell}.\rank_{\bit{1}}(r_k)\rangle$, just before the \textbf{return} of line \ref{line:return-level-l} in Algorithm \ref{alg:trie-intersection}.
Outputting this information is important for several applications, such as cases where set elements have satellite data associated to them.
For an element $x_j \in S_i$, the associated data $d_{j}$ is stored in array $D_i[1{..}n_i]$ such that $D[\rank(S_i, x_{j})] = d_{j}$. Typical applications are inverted indexes in IR (where ranking information, such as frequencies, is associated to inverted list elements), and the Leapfrog Triejoin algorithm \cite{Vicdt14} (where at each step we must compute the intersection of sets, and for each element in the intersection we must go down following a pointer associated to it). 

We have proved the following theorem:
\begin{theorem} \label{theor:compressed-TP}
Let $\mathcal{S} = \{S_1, \ldots, S_N\}$ be a family of $N$ integer sets, each of size $|S_i| = n_i$ and universe $[0{..}u)$. There exists a data structure able to represent each set $S_i$ using $2(\trie{S_i}-n_i+1) + o(\trie{S_i})$ bits, such that given a query $\mathcal{Q} = \{i_1, \ldots, i_k\} \subseteq [1{..}N]$, the intersection $\mathcal{I(Q)}= \cap_{i\in\mathcal{Q}}{S_i}$ can be computed in $\bigoh{k\delta\lg{(u/\delta)}}$ time, where $\delta$ is the alternation measure of $\mathcal{Q}$. Besides, for all $x \in \mathcal{I(Q)}$, the data structure also allows one to obtain the sequence $\langle \rank(S_{i_1}, x), \ldots, \rank(S_{i_k}, x)\rangle$ asymptotically for free.
\end{theorem}

\section{Compressing Runs of Elements} \label{sec:trie-runs}

Next, we exploit the maximal runs of a set $S$ to reduce the space usage of $\bintrie{S}$, as well as intersection time. Runs tend to form full subtrees in the corresponding binary tries. See, e.g., the full subtree whose leaves correspond to elements $8, 9, 10, 11$ in the binary trie of Figure \ref{fig:binary-tries-runs}. Let $v$ be a $\bintrie{S}$ node whose subtree is full. Let $\mathsf{depth}(v) = d$. If $b$ denotes the binary string (of length $d$) corresponding to node $v$, the $2^{\ell-d}$ leaves covered by $v$ correspond to the range of integers whose binary encodings are
$$b\cdot \bit{0}^{\ell - d}, b\cdot \bit{0}^{\ell - d-1}\bit{1}, 
b\cdot \bit{0}^{\ell - d-2}\bit{10}, \ldots, b\cdot \bit{1}^{\ell-d}.$$
So, the subtree of $v$ can be removed, keeping just $v$, saving space and still being able to recover the removed elements.

\begin{figure}[ht]
    \centering
\begin{tikzpicture}[
level distance=25mm,
level 1/.style={sibling distance=144mm},
level 2/.style={sibling distance=73mm},
level 3/.style={sibling distance=35mm},
level 4/.style={sibling distance=19mm},
thick, scale=0.32
%
]

\tikzstyle{c}  = [draw, thick,shape=circle]
\tikzstyle{cb} = [draw, thick, shape=circle, blue, fill=blue!10!white, text=blue]
\tikzstyle{r}  = [draw, thick,shape=rectangle,minimum width=1mm, red,fill=red!10!white]
\tikzstyle{b}  = [draw, thick,shape=rectangle,minimum width=1mm,blue,fill=blue!10!white]
\tikzstyle{tr} = [draw, thick,isosceles triangle, shape border rotate=90, anchor=north]

\node[c]{}[edge from parent]
    child {node[c] {}
        child {node[c]{}
           child { node[c]{}
              child { edge from parent [dotted];}
              child { node[r]{\footnotesize 1}
                  edge from parent coordinate (r0);
              }
              edge from parent coordinate (l0);
           }
           child { node[c]{}
               child { edge from parent [dotted];}
               child { node[r]{\footnotesize 3}
                   edge from parent coordinate (r1);
               }
               edge from parent coordinate (r2);
           }
           edge from parent coordinate (l1);
        }
        child { node[c]{}
            child { edge from parent [dotted];}
            child { node[c]{}
                child { edge from parent [dotted];}
                child { node[r]{\footnotesize 7}
                    edge from parent coordinate (r3);    
                }
                edge from parent coordinate (r4);
            }
            edge from parent coordinate (r5);
        }
        edge from parent coordinate (l2);
    }
    child {node[c]{}
        child {node[cb]{}
            child{ node[cb]{}
                child { node[b]{\footnotesize 8}
                    edge from parent [blue] coordinate (l3);
                }
                child { node[b]{\footnotesize 9}
                    edge from parent [blue] coordinate (r6);
                }
                edge from parent [blue] coordinate (l4);
            }
            child { node[cb]{}
                child { node[b]{\footnotesize 10}
                    edge from parent [blue ]coordinate (l44);
                }
                child { node[b]{\footnotesize 11}
                    edge from parent [blue] coordinate (r7);
                }
                edge from parent [blue] coordinate (r8);
            }
            edge from parent coordinate (l5);
        }
        child {node[c]{}
            child { node[c]{}
                child { node[r]{\footnotesize 12}
                    edge from parent coordinate (l6);
                }
                child { edge from parent [dotted];}
                edge from parent coordinate (l7);
            }
            child { edge from parent [dotted];}
            edge from parent coordinate (r9);
        }
        edge from parent coordinate (r10);
    };
\node at([xshift=-3mm,yshift=2mm]l0)  {\tt \footnotesize 0};
\node at([xshift=-3mm,yshift=2mm]l1)  {\tt \footnotesize 0};
\node at([xshift=-3mm,yshift=2mm]l2)  {\tt \footnotesize 0};
\node at([xshift=-3mm,yshift=2mm]l3)  {\tt \footnotesize 0};
\node at([xshift=-3mm,yshift=2mm]l4)  {\tt \footnotesize 0};
\node at([xshift=-3mm,yshift=2mm]l5)  {\tt \footnotesize 0};
\node at([xshift=-3mm,yshift=2mm]l6)  {\tt \footnotesize 0};
\node at([xshift=-3mm,yshift=2mm]l7)  {\tt \footnotesize 0};
\node at([xshift=-3mm,yshift=2mm]l44)  {\tt \footnotesize 0};
\node at([xshift=3mm,yshift=2mm]r0) {\tt \footnotesize 1};
\node at([xshift=3mm,yshift=2mm]r1) {\tt \footnotesize 1};
\node at([xshift=3mm,yshift=2mm]r2) {\tt \footnotesize 1};
\node at([xshift=3mm,yshift=2mm]r3) {\tt \footnotesize 1};
\node at([xshift=3mm,yshift=2mm]r4) {\tt \footnotesize 1};
\node at([xshift=3mm,yshift=2mm]r5) {\tt \footnotesize 1};
\node at([xshift=3mm,yshift=2mm]r6) {\tt \footnotesize 1};
\node at([xshift=3mm,yshift=2mm]r7) {\tt \footnotesize 1};
\node at([xshift=3mm,yshift=2mm]r8) {\tt \footnotesize 1};
\node at([xshift=3mm,yshift=2mm]r9) {\tt \footnotesize 1};
\node at([xshift=3mm,yshift=2mm]r10) {\tt \footnotesize 1};
\end{tikzpicture}

\vspace{0.75cm}

\begin{tikzpicture}[
    scale=0.35, node distance=0.4cm
]

\begin{scope}[every node/.style={draw, thick,shape=rectangle,minimum width=1mm, black,fill=red!10!white}]
\node[] (A) {\texttt{11}};
\node[left = of A, draw=white, fill = white] {$B_1:$};

\node[below = of A] (B) {\texttt{11}};
\node[right = of B] (C) {\texttt{11}};
\node[left = of B, draw=white, fill = white] {$B_2:$};

\node[below = of B] (D) {\texttt{11}};
\node[right = of D] (E) {\texttt{01}};
\node[right = of E, text=blue, shape=rectangle, blue, fill=blue!10!white] (F) {\texttt{00}};
\node[right = of F] (G) {\texttt{10}};
\node[left = of D, draw=white, fill = white] {$B_3:$};

\node[below = of D] (H) {\texttt{01}};
\node[right = of H] (I) {\texttt{01}};
\node[right = of I] (J) {\texttt{01}};
\node[below = of G] (M) {\texttt{10}};
\node[left = of H, draw=white, fill = white] {$B_4:$};

\end{scope}

\begin{scope}
\draw[dotted, thick] (A)  to node {} (B);
\draw[dotted, thick] (A)  to node {} (C);

\draw[dotted, thick] (B)  to node {} (D);
\draw[dotted, thick] (B)  to node {} (E);

\draw[dotted, thick] (C)  to node {} (F);
\draw[dotted, thick] (C)  to node {} (G);

\draw[dotted, thick] (D)  to node {} (H);
\draw[dotted, thick] (D)  to node {} (I);

\draw[dotted, thick] (E)  to node {} (J);


\draw[dotted, thick] (G)  to node {} (M);
\end{scope}

\end{tikzpicture}

\caption{Above, the binary trie representing set $\{1, 3, 7, 8, 9, 10, 11, 12\}$. Notice that the subtree whose leaves correspond to elements $8, 9, 10, 11 $ is a full subtree. Below, our compact representation removing full subtrees and encoding their roots with \bit{00}.}
\label{fig:binary-tries-runs}
\end{figure}
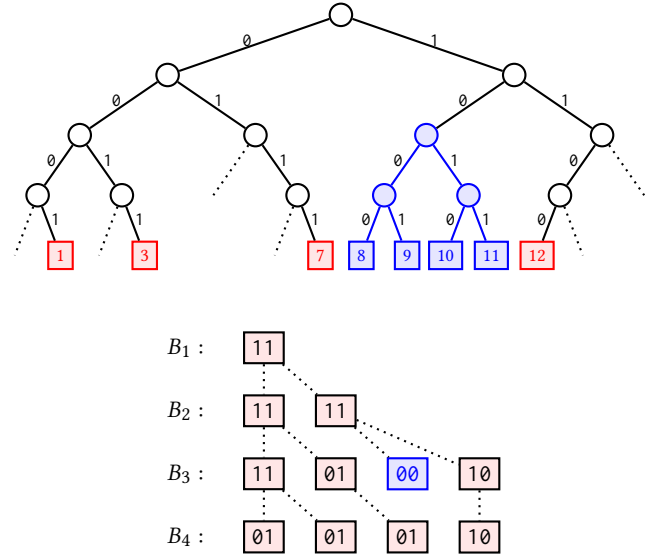

\begin{definition} \label{def:rTrie}
Let $S \subseteq [0{..}u)$ be a set of $n$ elements. We define $\trieRun{S}$ as the number of edges in $\bintrie{S}$ after removing the maximal full subtrees.  
\end{definition}
Notice this immediately implies $\trieRun{S} \le \trie{S} \le 2\gap{S}$.
We can, however, prove tighter bounds.
Assume that $S$ can be partitioned into the maximal runs $R_1, \ldots, R_r$, each of $\ell_i = |R_i|$ elements. For a given $R_i$, notice that its elements correspond to $\ell_i$ contiguous leaves in $\bintrie{S}$. According to Lemma \ref{lemma:bintrie-properties} (item 1), these $R_i$ contiguous leaves are covered by at most $2\lfloor\lg{(\ell_i/2)}\rfloor$ cover nodes. This is the case that removes the least edges, so we analyze it.  Among the cover nodes, there are 2 whose subtree has $0$ edges, 2 whose subtree has 2 edges, 2 whose subtree has 6 edges, and so on. In general, for each $i = 1, \ldots, \lfloor\lg{(\ell_i/2)}\rfloor$, there are 2 cover nodes whose subtree has $2^i-2$ edges. If we remove all these subtrees, the total number of nodes removed is 
$2\sum_{i=i}^{\lfloor\lg{(\ell_i/2)}\rfloor}{(2^{i}-2)} = 2\ell_i - 4\lg{\ell_i}$.
Since this is the case that removes the least edges belonging to full subtrees, we can bound
\begin{equation}\label{eq:rtrie}
\trieRun{S} \le \trie{S} - \sum_{i=1}^{r}{(2\ell_i - 4\lg{\ell_i})}.
\end{equation}
We can also prove the following bounds.
\begin{lemma} \label{lemma:trierun}
Given a set $S\subseteq [0{..}u)$ of $n$ elements, it holds that
\begin{enumerate}
\item $\trieRun{S} \le 2\cdot\min{\{\rle{S} + \sum_{i=1}^{r}{\lg{\ell_i}}, \gap{S}\}}$. \label{item:trierun-1}
\item $\exists a \in [0{..}u)$, such that $$\trieRun{S+a} \le \min{\{\rle{S} - \sum_{i=1}^{r}{\ell_i} + 3\sum_{i=1}^{r}{\lg{\ell_i}},\gap{S}\}} + 2n -2.$$ \label{item:trierun-2}
\item $\trieRun{S+a} \le \min{\{\rle{S}  - \sum_{i=1}^{r}{\ell_i} + 3\sum_{i=1}^{r}{\lg{\ell_i}},\gap{S}\}} + 2n -2$ on average, assuming $a \in [0{..}u)$ is chosen uniformly at random. \label{item:trierun-3}
\end{enumerate}
\end{lemma}
\begin{proof}
As $S$ has maximal runs $R_1, \ldots, R_r$, each of $\ell_i = |R_i|$ elements, 
notice that we can rewrite $$\gap{S} = \sum_{i=1}^{r}{(\lfloor\lg{(z_i-1)}\rfloor + 1)} + \sum_{i=1}^{r}{(\ell_i - 1)}.$$  
Since $\trieRun{S} \le \trie{S} \le 2\gap{S}$, and $\trieRun{S} \le \trie{S} - \sum_{i=1}^{r}{(2\ell_i - 4\lg{\ell_i})}$ (Equation \ref{eq:rtrie}), it holds that
\begin{align*} 
\trieRun{S} & \le  \trie{S} - \sum_{i=1}^{r}{(2\ell_i - 4\lg{\ell_i})} \\
            & \le  2(\sum_{i=1}^{r}{(\lfloor\lg{(z_i-1)}\rfloor + 1)} \\
            &      + \sum_{i=1}^{r}{(\ell_i - 1)}) - \sum_{i=1}^{r}{(2\ell_i - 4\lg{\ell_i})} \\
            & =    2\sum_{i=1}^{r}{(\lfloor\lg{(z_i-1)}\rfloor + 1)} + 4\sum_{i=1}^{r}{\lg{\ell_i}} \\
            & \approx  2(\rle{S} + \sum_{i=1}^{r}{\lg{\ell_i}}). 
\end{align*}

This proves item 1. The remaining items can be proved similarly from items 2 and 3 of Lemma \ref{lemma:trie}.
\end{proof}

In our compact representation, we encode a cover node whose full subtree has been removed using $\bit{00}$. Recall that $\bit{00}$ is an invalid node encoding, hence we use it now as a special mark. See Figure \ref{fig:binary-tries-runs} for an illustration. It remains now to explain how to carry out operations on this representation.

\subsection{\texorpdfstring{$\rank(S, x)$}{rank}}

We proceed mostly as explained for our original representation, traversing the path from the root to the leaf representing $y = \predecessor(S, x)$. However, this time there can be subtrees that have been removed as explained, hence $\rank(S, x)$ not necessarily corresponds to the $\rank$ up to the $\bit{1}$ we arrive at the last level $B_\ell$: this $\rank$ only gives partial information to compute the operation. It remains to account for all removed $\bit{1}$s corresponding to leaves of complete subtries. To do so, at each level of the trie we must regard the $\bit{00}$s that lie to the left of the path we are following. Notice that every $\bit{00}$ at depth $1\le l \le \ell$ corresponds to a full subtrie of $2^{\ell-l+1}$ leaves that have been removed from $B_\ell$. 
To account for them, we keep a variable $d$ during the traversal, initialized as $d \gets 0$. At each level $l\ge 1$, being at position $i$ of $B_l$, we carry out $d \gets d + 2^{\ell-l+1}\cdot B_{l}.\rank_{00}(i)$, where $\rank_{00}$ is the number of nodes encoded $\bit{00}$ before position $i$ in $B_{l}$. Notice this is different to the number of $\bit{00}$s before position $i$, in our case we need to count the number of $\bit{00}$s that are aligned with odd positions (and hence represent removed nodes). 




\subsection{Intersection}

Given a query $\mathcal{Q} = \{i_1, \ldots, i_k\} \subseteq [1{..}N]$, the procedure is similar to that of Algorithm \ref{alg:trie-intersection}. The only difference is that if in a given trie $\bintrie{S_i}$ we arrive at a node encoded $\bit{00}$, we can temporarily exclude that trie from the intersection without affecting the result. The rationale is that every possible element below that node belongs to $S_i$, hence the intersection within that subtree is independent of $S_i$ and we can temporarily exclude it. To implement this idea, we keep boolean flags $f_1, \ldots, f_k$ such that $f_j$ corresponds to $\bintrie{S_{i_j}}$. Initially, we set $f_i \gets \true$, for $1\le i\le k$.  
If, during the intersection process, we arrive at a node encoded \bit{00} in $\bintrie{S_i}$, we set $f_i \gets \false$. The idea is that at each node visited during the recursive procedure, only the tries whose flag is $\true$ participate in the intersection. The remaining ones are within a full subtrie, so they are currently excluded. When the recursion at a node encoded \bit{00} in $\bintrie{S_i}$ finishes, we set $f_i \gets \true$ again.
If, at a given point, all tries have been temporarily excluded but one, let us say $\bintrie{S_j}$, we only need to traverse the current subtree in $S_j$, copying it verbatim to the output. If this subtree contains nodes encoded $\bit{00}$, they will appear in the output. This way, the maximal runs of successive elements in the output will be covered by nodes encoded $\bit{00}$. This fact is key for the adaptive running time of our algorithm, as we shall see below.

\subsection{Running Time Analysis}

Next we analyze the running time of our intersection algorithm on $\trieRun{S}$ compressed sets. As we have seen, runs can be exploited to use less space. Now, we also prove that runs can be exploited to improve intersection computation time. The rationale is that if we intersect sets with runs of successive elements, very likely the output will have some runs too. We will show that our algorithm is able to provide a smaller partition certificate when there are runs in the output. Next, we redefine the concept of partition certificate \cite{BKtalg08}, taking into account the existence of runs in the output.
\begin{definition}
\label{def:certificate-runs}
Given a query $\mathcal{Q} = \{i_1, \ldots, i_k\} \subseteq [1{..}N]$,
a run-partition certificate for it is a partition of the universe 
$[0{..}\univSize)$ into a set of intervals $\mathcal{P} = \{I_1, I_2, \ldots, I_p\}$, 
such that the following conditions hold:
\begin{enumerate}
\item $\forall x \in \mathcal{I(Q)}, \exists I_j \in \mathcal{P}$, $x\in I_j \wedge \mathcal{I(Q)}\cap I_j = I_j$; 
\label{item:1-runs}

\item $\forall x \not \in \mathcal{I(Q)}, \exists I_j \in \mathcal{P}, x \in I_j~\wedge~\exists q \in \mathcal{Q}, S_q \cap I_j = \emptyset$.
\label{item:2-runs}

\end{enumerate}
\end{definition}
The second item is the same as Barbay and Kenyon's partition certificate, and correspond to intervals that cover the elements not in the intersection. The first item, on the other hand, correspond to the elements in the intersection. Unlike Barbay and Kenyon, in our partition certificate elements in the intersection are not necessarily covered by singletons: if there is a run of successive elements in the intersection, our definition allows for them to be covered by a single interval in the certificate. 

\begin{definition}
For a given query $\mathcal{Q} = \{i_1, \ldots, i_k\} \subseteq [1{..}N]$, let $\xi$ denote the size of the smallest run-partition certificate of $\mathcal{Q}$.
\end{definition}
It is easy to see that $\xi \le \delta$ holds. Besides, although $|\mathcal{I(Q)}| \le \delta$ holds, in our case there can be query instances such that $\xi < |\mathcal{I(Q)}|$.
Figure \ref{fig:partition-runs} illustrates our definition for an intersection of 4 sets on the universe $[0{..}15)$. Notice that $\xi = 5$ for this example, whereas $|\mathcal{I(Q)}| = 6$ and $\delta = 9$.

\begin{figure}
\begin{tabular}{ccccc|cc|c|cccc|cccccc|cc}
$S_{i_1}:$ &   &   &  & 7 & 8 & 9  & 10 & 11 & 12 & 13 & 14 & 15 \\
$S_{i_2}:$ &  & 5 & 6 & 7 & 8 & 9  & 10 & 11 & 12 & 13 & 14 &   \\
$S_{i_3}:$ & 4 & 5 & 6 & 7 & 8 & 9 &    & 11 & 12 &  13 & 14   &  \\
$S_{i_4}:$ &  &  &  &  & 8 & 9     & 10 & 11 & 12 & 13 & 14 & 15 \\
\end{tabular}
\caption{A query $\mathcal{Q} = \{S_{i_1}, S_{i_2}, S_{i_3}, S_{i_4}\}$ and its smallest run-partition certificate
$\mathcal{P} = \{[0{..}7], \linebreak[3] [8{..}9], \linebreak[3] [10{..}10], \linebreak[3] [11{..}14], \linebreak[3] [15{..}15]\}$  
of size $\xi = 5$.}
\label{fig:partition-runs}
\end{figure}


Our main result is stated in the following theorem:
\begin{theorem} \label{theor:run-compressed-TP}
Let $\mathcal{S} = \{S_1, \ldots, S_N\}$ be a family of $N$ integer sets, each of size $|S_i| = n_i$ and universe $[0{..}u)$. There exists a data structure able to represent each set $S_i$ using $2\trieRun{S_i} (1+ o(\trieRun{S_i}))$ bits, such that given a query $\mathcal{Q} = \{i_1, \ldots, i_k\} \subseteq [1{..}N]$, the intersection $\mathcal{I(Q)}= \cap_{i\in\mathcal{Q}}{S_i}$ can be computed in $\bigoh{k\xi\lg{(u/\delta)}}$ time, where $\xi$ is the run alternation measure of $\mathcal{Q}$. 
\end{theorem}
\begin{proof}
Consider the smallest partition certificate $\mathcal{P}$ of query $\mathcal{Q}$, which partitions the universe $[0{..}u)$ into $\xi$ intervals $I_1, \ldots, I_{\xi}$, of size $L_1, \ldots, L_{\xi}$, respectively. Let us cover the $\xi$ intervals with trie nodes, in order to bound the number of nodes of $\certificate{\mathcal{Q}}$ (and, hence, the running time). As we already saw in the proof of Theorem \ref{theor:TP-adaptive}, all intervals $I_i$ such that $I_i\cap\mathcal{I(Q)} = \emptyset$ are covered by at most $\bigoh{\lg{L_i}}$ nodes in $\certificate{\mathcal{Q}}$. We now prove the same for intervals $I_j \subseteq \mathcal{I(Q)}$. The only thing to note is that our algorithm stops as soon as it arrives to a node covering successive elements in the output. As there can be $\bigoh{\lg{L_j}}$ such cover nodes, we can prove that $[0{..}u)$ can be covered by $\bigoh{\sum_{i=1}^{\xi}{\lg{L_i}}} = \bigoh{\sum_{i=1}^{\xi}{\lg{\!(u/\xi)}}}$ $\certificate{\mathcal{Q}}$ nodes, hence $\certificate{\mathcal{Q}}$ has $\bigoh{\bigoh{\xi\lg{\!(u/\xi)}}}$ nodes overall. The result follows from the fact that at each node the time is $\bigoh{k}$. 
\end{proof}

%% file: 04-implementation.tex
\section{Implementation}

\subsection{Implementing the Tries}

We implemented bit vectors $B_1, \ldots, B_\ell$ in plain form using class \texttt{bit\_vector<>} from the \texttt{sdsl} library \cite{GPspe14}. We support $\rank_{\bit{1}}$ on them using different data structures and obtain:
\begin{description}
\item[\texttt{trie v}, \texttt{rTrie v}:] the variants defined in Section \ref{sec:compressed-trie} and \ref{sec:trie-runs}, respectively, using \texttt{rank\_support\_v} for $\rank_{\bit{1}}$. It uses 25\% extra space on top of the bit vector, and supports $\rank_{\bit{1}}$ in $\bigoh{1}$ time. 
\item[\texttt{trie v5}, \texttt{rTrie v5}:] use \texttt{rank\_support\_v5}, requiring 6.25\% extra space on top of the bit vectors, and supporting $\rank_{\bit{1}}$ in $\bigoh{1}$ time. This alternative uses less space, yet it is slower in practice. 

\item[\texttt{trie IL}, \texttt{rTrie IL}:] use  \texttt{rank\_support\_il}, aiming at reducing the number of cache misses to compute $\rank_{\bit{1}}$. We use block size $512$, requiring 12.5\% extra space on top of the bit vectors, while supporting $\rank_{\bit{1}}$ in constant time.  
\end{description}
We do not store any $\rank_{\bit{1}}$ data structure for the last-level bit vector $B_\ell$. So, we are not able to compute $\rank(S, x)$ for a given set $S$ (recall that this operation is equivalent to a $\rank_{\bit{1}}$ on the corresponding position of $B_\ell$). This is in order to be fair, as most state-of-the-art alternatives we tested do not support this operation. 
Also, we do not use data structures for $\select_{\bit{1}}$ in our implementation, as we only test intersections. 

\subsection{Implementing the Intersection Algorithm}

We implemented Algorithm \ref{alg:trie-intersection} on our compressed trie representation, as well as the variant that eliminates subtrees that are full because of runs. We follow the descriptions from Sections \ref{sec:compressed-trie} and \ref{sec:trie-runs} very closely. We implemented, however, two alternatives for representing the output: (1) the binary trie representation, and (2) the plain array representation. In our experiments we will use the latter, to be fair: all testes alternatives produce their outputs in plain form.

\subsection{A Multithreaded Implementation}

Our intersection algorithm can be implemented in a multithreading architecture quite straightforwardly. Let $t$ denote the number of available threads. Then, we define $c = \lfloor\lg{t}\rfloor$. Hence, our algorithm proceeds as in Algorithm \ref{alg:trie-intersection}, generating a binary trie of height $c$ (that we will call \emph{top trie}), with at most $t$ leaves. Then, we execute again Algorithm \ref{alg:trie-intersection}, this time in parallel with each thread starting from a different leaf of the top trie. If there are less than $t$ leaves in the top trie, we go down until the $t$ threads have been allocated. 
Each thread generates its own output in parallel, using our compact trie representation. Once all threads finish, we concatenate these tries to generate the final output. We just need to count, in parallel, how many nodes there are in each level of the trie. Then, we allocate a bit vector of the appropriate size for each level, and each thread will write its own part of the output in parallel. This is just a simple approach that does not guarantees load balancing among threads, however it works relatively well in practice.


\subsection{Source Code Availability}

Our source code and instructions to replicate our experiments are available at \url{https://github.com/jpcastillog/compressed-binary-tries}.

%% file: 05-experimental-results.tex
\section{Experimental Results}

Next, we experimentally evaluate our approaches.

\subsection{Experimental Setup}

\paragraph{Hardware and Software}
All experiments were run on server with an i7 10700k CPU with 8 cores and 16 threads, with disable turbo boost running at $4.70$ GHz in all cores and running Ubuntu $20.04$ LTS operating system. We have $32$GB of RAM(DDR4-$3.6$GHz) running in dual channel. Our implementation is developed in C++, compiled with \texttt{g++} $9.3.0$ and optimization flags \texttt{-O3} and \texttt{-march=native}. 

\paragraph{Datasets and Queries.} 
In our tests, we used family of sets corresponding to inverted indexes of three standard document collections: Gov2 \footnote{\url{https://www-nlpir.nist.gov/projects/terabyte/}}, ClueWeb09 \footnote{\url{https://lemurproject.org/}}, and CC-News \cite{ccnews2020}. 
For Gov2 and ClueWeb09 collections, we used the freely-available inverted indexes and query log by D.~Lemire \footnote{See \url{https://lemire.me/data/integercompression2014.html} for download details.}, corresponding to the URL-sorted document enumeration \cite{Silvestri07}, which tends to yield runs of successive elements in the sets. 
The query log contains 20{,}000 random queries from the TREC million-query track (1MQ). Each query has at least 2 query terms. Also, each term is in the top-1M most frequently queried terms. 
For CC-News we use the freely-available inverted index by Mackenzie et al.~\cite{ccnews2020} in a \emph{Common Index File Format} (CIFF) \cite{ciff}, as well as their query log of 9{,}666 queries.
Table \ref{table:datasets-summary} shows a summary of statistics of the collections. 
In all cases, we only keep inverted lists with length at least 4096. 

\begin{table}[ht]
\caption{Dataset summary and average space usage (in bits per integer, bpi) for different compression measures and baseline representations.}
\label{table:datasets-summary}
\centering
\begin{tabular}{lrrr}
\toprule
                                        & \textsf{Gov2} & \textsf{ClueWeb09} & \textsf{CC-News}\\ 
\midrule
\# Lists      & 57,225        & 131,567 & 79,831                                  \\ 
\# Integers   & 5,509,206,378 & 14,895,136.282  & 18,415,151,585              \\ 
$u$        & 25,205,179    & 50,220,423   & 43,495,426                             \\ 
$\lceil\lg{u}\rceil$   & 25            & 26 & 26    \\ 
\midrule 
$\gap{\set}$        & 2.25          & 3.25      & 3,70          \\ 
$\rle{\set}$        & 1.99          & 3.33      & 4,23          \\ 
$\trie{\set}$       & 3.48          & 4.56      & 5,18          \\ 
$\trieRun{\set}$    & 2.51          & 4.00      & 5,12           \\
\midrule
Elias $\gamma$      & 3.71          & 5.74      & 6.81          \\
Elias $\delta$      & 3.64          & 5.40      & 6.69          \\
Fibonacci           & 3.90          & 5.35      & 6.09          \\
Elias $\gamma$ 128 & 4.07 &	6.10    & 7.05\\
Elias $\delta$ 128 & 4.00 & 5.77    & 7.17\\
Fibonacci 128      & 4.26 &	5.71    & 6.45\\
\texttt{rrr\_vector<>}   & 11.82 & 19.94  & 11.29\\
\texttt{sd\_vector<>}    & 8.45 & 8.52     & 7.17 \\
\bottomrule
\end{tabular}
\end{table}

\paragraph{Compression Measures and Baselines.}
Table \ref{table:datasets-summary} also shows the average bit per integer (bpi) for different compression measures on our tested set collections. 
We also show the average bpi for different integer compression approaches, namely Elias $\gamma$ and $\delta$ \cite{Elias1975}, Fibonacci \cite{FKdam96}, \textsf{rrr\_vector<>} \cite{RRR07}, and \texttt{sd\_vector<>} \cite{OSalenex07}, all of them from the \texttt{sdsl} library \cite{GPspe14}. In particular, Elias $\gamma$, $\delta$, and Fibonacci codes are know for yielding highly space-efficient set representations in IR \cite{BCC2010}, hence they are a strong baseline for comparison. We show a plain version of them, as well as variants with blocks of 128 integers. The latter are needed to speed up decoding. However, it is well known that these codes are relatively slow to be decoded, and hence yield higher intersection times.  
On the other hand, \texttt{sd\_vector<>} uses $n\lg{(u/n)} + 2n + o(n)$ bits to encode a set of $n$ elements from the universe $[0{..}u)$, which is close to the worst-case upper bound $\lowerB$. Finally, $\texttt{rrr\_vector<>}$ uses $\lowerB + o(u)$ bits of space. As it can be seen, the $o(u)$-bit term yields a space usage higher than the remaining alternatives.
These values will serve as a baseline to compare our results.

\paragraph{Indexes Tested.}
We compare our proposal with state-of-the-art set compression approaches. In particular, we use the code available at the project \emph{Performant Indexes and Search for Academia}\footnote{\url{https://github.com/pisa-engine/pisa}} (PISA) \cite{MSMS2019} for the following approaches:
\begin{description}
    \item[\texttt{IPC}:] the Binary Interpolative Coding approach by Moffat et al.~\cite{MSir00}. This is a highly space-efficient approach, with a relatively slow processing performance \cite{BCC2010,MSir00}.
    \item[\texttt{PEF Opt}:] the highly competitive approach by Ottaviano and Venturini \cite{OVsigir14}, which partitions the universe into non-uniform blocks and represents each block appropriately. 
    \item[\texttt{OptPFD}:] The Optimized PForDelta approach by Yan et al.~\cite{YDSwww09}.
    \item[\texttt{SIMD-BP128}:] The highly efficient approach by Lemire and Boytsov \cite{LBspe15}, aimed at decoding billions of integers per second using vectorization capabilities of modern processors.  
    \item[\texttt{Simple16}:] The approach by Zhang at al.~\cite{ZLSwww08}, a variant of the \texttt{Simple9} approach \cite{AMir05} that combines a relatively good space usage and an efficient intersection time. 
    \item[\texttt{Varint-G8IU}:] The approach by Stepanov et al.~\cite{SGREOcikm11}, using SIMD instructions to speed up set manipulation.
    \item[\texttt{VarintGB}:] The approach presented by Dean \cite{Dwsdm09}.
\end{description}    
We also compared with the following approaches, available from their authors:
\begin{description}
    \item[\texttt{Roaring}:] the compressed bitmap approach by Lemire et al.~\cite{LKKDOSKspe18}, which is widely used as an indexing tool on several systems and platforms \footnote{See, e.g., \href{https://roaringbitmap.org/}{https://roaringbitmap.org/}}. Roaring bitmaps are highly competitive, taking full advantage of modern CPU hardware architectures. We use the C++ code from the authors \footnote{\href{https://github.com/RoaringBitmap/CRoaring}{https://github.com/RoaringBitmap/CRoaring}.}
    
    \item[\texttt{RUP}:] The recent recursive universe partitioning approach by Pibiri \cite{Pdcc21}, using also SIMD instructions to speed up processing. We use the code from the author \footnote{\href{https://github.com/jermp/s_indexes}{https://github.com/jermp/s\_indexes}}.
\end{description}

\subsection{Experimental Intersection Queries}

Table \ref{table:experimental-results} shows the average experimental intersection time and space usage (in bits per integer) for all the alternatives tested. 
\begin{table}
\caption{Average intersection time and space usage (in bits per integer) for all alternatives tested.}
\label{table:experimental-results}
\centering
\begin{tabular}{lrrrrrrrr}
\toprule
  & \multicolumn{2}{c}{\textsf{Gov2}} &  \multicolumn{2}{c}{\textsf{ClueWeb09}} & \multicolumn{2}{c}{\textsf{CC-News}} \\
\cline{2-3} \cline{4-5} \cline{6-7} 
Data Structure            &  Space  &  Time  & Space  &  Time & Space  &  Time\\
\midrule
\texttt{IPC}              & 3.34  &    8.66    & 5.15 &   30.18  & 5.87	& 68.98\\ 
\texttt{Simple16}         & 4.65  &    2.44    & 6.72 &   8.66  & 6.88	& 19.74\\
\texttt{OptPFD}           & 4.07  &    2.15    & 6.28 &   7.79 & 6.50   & 11.80   \\ 
\texttt{PEF Opt}          & 3.62  &    1.88    & 5.85 &   6.50 & 5.80   & 17.33   \\ 
\texttt{VarintGB}         & 10.80 &    1.43    & 11.40&   7.34 & 11.04  & 12.38 \\
\texttt{Varint-G8IU}      & 9.97  &    1.38    & 10.55&   5.25 & 10.24  & 12.09\\
\texttt{SIMD-BP128}       & 6.07  &    1.29    & 8.98 &   4.47 & 7.36   & 15.90\\ 
\texttt{Roaring}          & 8.77  &    1.09    & 12.62&   3.75 & 9.86   & 5.56\\
\texttt{RUP}              & 5.04  &   1.10    & 8.44 &    4.27  & 8.41  & 5.44 \\   
\midrule
\texttt{trie} (\texttt{v5})  & 5.18 &  1.21  & 7.46 & 2.81  & 8.77 & 8.72\\
\texttt{trie} (\texttt{IL})  & 5.41 &  1.06  & 7.83 &  2.42 & 9.30  & 7.46\\
\texttt{trie} (\texttt{v})   & 5.85 &  0.77  & 8.50 & 1.64  & 9.99 & 5.21\\
\midrule
\texttt{rTrie} (\texttt{v5})  & 4.22 &  1.22 & 6.95 & 3.07 & 8.73 & 9.74\\
\texttt{rTrie} (\texttt{IL})  & 4.42 &  1.10 & 7.31 &  2.62 & 9.16 & 8.13\\
\texttt{rTrie} (\texttt{v})   & 4.81 &  0.77 & 7.96 &  1.96 & 9.95 & 6.09\\
\bottomrule
\end{tabular}
\end{table}
Figure \ref{fig:experimental-query-time} shows the same results, using space vs.~time plots.
\begin{figure}[ht]
     \centering
        \begin{tikzpicture}[scale=0.72]
        \begin{axis}[
        legend columns = 3,
        title = {\sf Intersection Queries, Gov2},
        xlabel={Space [bpi]},  
        legend style={nodes={scale=0.8, transform shape,right},
        },
        ylabel={Intersection time [millisecs]},
        xmin=0, xmax=11.2,
        ymin=0, ymax=9
        ]
        \addplot+[
            mark size = 4pt,
            only marks,
            color=black,
            mark=pentagon
        ]
        coordinates {(3.34, 8.66)};
        
        \addplot+[
            mark size = 4pt,
            only marks,
            color=black,
            mark=square
        ]
        coordinates {(4.65, 2.44)};
        
        \addplot+[
            mark size = 4pt,
            only marks,
            color=black,
            mark=o
        ]
        coordinates {(4.07, 2.15)};
        
        \addplot+[
            mark size = 4pt,
            only marks,
            color=black,
            mark=diamond
        ]
        coordinates {(3.62, 1.88)};
        
        \addplot+[
            mark size = 4pt,
            only marks,
            color=black,
            mark=otimes
        ]
        coordinates {(10.80, 1.43)};
        
        \addplot+[
            mark size = 4pt,
            only marks,
            color=black,
            mark=oplus
        ]
        coordinates {(9.97, 1.38)};
        
        \addplot+[
            mark size = 4pt,
            only marks,
            color=black,
            mark=x
        ]
        coordinates {(6.07, 1.29)};
        
        \addplot+[
            mark size = 4pt,
            only marks,
            color=black,
            mark=+
        ]
        coordinates {(8.77, 1.09)};
        
        \addplot+[
            mark size = 4pt,
            only marks,
            color=black,
            mark=star
        ]
        coordinates {(5.04, 1.10)};
        
        \addplot+[
            mark size = 4pt,
            only marks,
            color=red,
            mark=triangle
        ]
        coordinates {(4.22, 1.22)};
        
        \addplot+[
            mark size = 4pt,
            only marks,
            color=blue,
            mark=triangle
        ]
        coordinates {(4.42, 1.10)};
        
        \addplot+[
            mark size = 4pt,
            only marks,
            color=teal,
            mark=triangle
        ]
        coordinates {(4.81, 0.77)};

        \addplot+[
            mark size = 4pt,
            only marks,
            color=red,
            mark=triangle,
            every mark/.append style={rotate=180}
        ]
        coordinates {(5.18, 1.21)};
        
        \addplot+[
            mark size = 4pt,
            only marks,
            color=blue,
            mark=triangle,
            every mark/.append style={rotate=180}
        ]
        coordinates {(5.41, 1.06)};
        
        \addplot+[
            mark size = 4pt,
            only marks,
            color=teal,
            mark=triangle,
            every mark/.append style={rotate=180}
        ]
        coordinates {(5.85, 0.77)};
        \end{axis}
    \end{tikzpicture}
    \hfill
        \begin{tikzpicture}[scale=0.72]
        \begin{axis}[
        legend columns = 3,
        title = {\sf Intersection Queries, ClueWeb09},
        xlabel={Space [bpi]},
        legend style={nodes={scale=0.8, transform shape, right}},
        ylabel={Intersection time [millisecs]},
        xmin=0, xmax=13,
        ymin=0, ymax=32
        ]
        \addplot+[
            mark size = 4pt,
            only marks,
            color=black,
            mark=pentagon
        ]
        coordinates {(5.15, 30.18)};
        
        \addplot+[
            mark size = 4pt,
            only marks,
            color=black,
            mark=square
        ]
        coordinates {(6.72, 8.66)};
        
        \addplot+[
            mark size = 4pt,
            only marks,
            color=black,
            mark=o
        ]
        coordinates {(6.28, 7.79)};
        
        \addplot+[
            mark size = 4pt,
            only marks,
            color=black,
            mark=diamond
        ]
        coordinates {(5.85, 6.50)};
        
        \addplot+[
            mark size = 4pt,
            only marks,
            color=black,
            mark=otimes
        ]
        coordinates {(11.40, 7.34)};
        
        \addplot+[
            mark size = 4pt,
            only marks,
            color=black,
            mark=oplus
        ]
        coordinates {(10.55, 5.25)};
        
        \addplot+[
            mark size = 4pt,
            only marks,
            color=black,
            mark=x
        ]
        coordinates {(8.98, 4.37)};
        
        \addplot+[
            mark size = 4pt,
            only marks,
            color=black,
            mark=+
        ]
        coordinates {(12.62, 3.75)};
        
        \addplot+[
            mark size = 4pt,
            only marks,
            color=black,
            mark=star
        ]
        coordinates {(8.44, 4.27)};
        
        \addplot+[
            mark size = 4pt,
            only marks,
            color=red,
            mark=triangle
        ]
        coordinates {(6.95, 3.07)};
        
        \addplot+[
            mark size = 4pt,
            only marks,
            color=blue,
            mark=triangle
        ]
        coordinates {(7.31, 2.62)};
        
        \addplot+[
            mark size = 4pt,
            only marks,
            color=teal,
            mark=triangle,
        ]
        coordinates {(7.96, 1.96)};
        
        \addplot+[
            mark size = 4pt,
            only marks,
            color=red,
            mark=triangle,
            every mark/.append style={rotate=180}
        ]
        coordinates {(7.46, 2.81)};
        
        \addplot+[
            mark size = 4pt,
            only marks,
            color=blue,
            mark=triangle,
            every mark/.append style={rotate=180}
        ]
        coordinates {(7.83, 2.42)};

        \addplot+[
            mark size = 4pt,
            only marks,
            color=teal,
            mark=triangle,
            every mark/.append style={rotate=180}
        ]
        coordinates {(8.50, 1.64)};
        \end{axis}
    \end{tikzpicture}
        \begin{tikzpicture}[scale=0.72]
        \centering
        \begin{axis}[
        legend columns = 3,
        title = {\sf Intersection Queries, CC-News},
        xlabel={Space [bpi]},  
        legend style={
            nodes={scale=0.8, transform shape,right},
            at={(0.5,-0.5)},
            anchor=south,
            legend columns=3,
            column sep=0.2cm
        },
        ylabel={Intersection time [millisecs]},
        xmin=0, xmax=11.5,
        ymin=0, ymax=72.0
        ]
        \addplot+[
            mark size = 4pt,
            only marks,
            color=black,
            mark=pentagon
        ]
        coordinates {(5.87, 68.98)};
        \addlegendentry{IPC}
        
        \addplot+[
            mark size = 4pt,
            only marks,
            color=black,
            mark=square
        ]
        coordinates {(6.88, 19.74)};
        \addlegendentry{Simple16}
        
        \addplot+[
            mark size = 4pt,
            only marks,
            color=black,
            mark=o
        ]
        coordinates {(6.50, 11.80)};
        \addlegendentry{OptPFD}
        
        \addplot+[
            mark size = 4pt,
            only marks,
            color=black,
            mark=diamond
        ]
        coordinates {(5.80, 17.33)};
        \addlegendentry{PEF Opt}
        
        \addplot+[
            mark size = 4pt,
            only marks,
            color=black,
            mark=otimes
        ]
        coordinates {(11.04, 12.38)};
        \addlegendentry{VarintGB}
        
        \addplot+[
            mark size = 4pt,
            only marks,
            color=black,
            mark=oplus
        ]
        coordinates {(10.24, 12.09)};
        \addlegendentry{Varint-G8IU}
        
        \addplot+[
            mark size = 4pt,
            only marks,
            color=black,
            mark=x
        ]
        coordinates {(7.36, 15.90)};
        \addlegendentry{SIMD-BP128}
        
        \addplot+[
            mark size = 4pt,
            only marks,
            color=black,
            mark=+
        ]
        coordinates {(9.86, 5.56)};
        \addlegendentry{Roaring}
        
        \addplot+[
            mark size = 4pt,
            only marks,
            color=black,
            mark=star
        ]
        coordinates {(8.41, 5.44)};
        \addlegendentry{RUP}
        
        \addplot+[
            mark size = 4pt,
            only marks,
            color=red,
            mark=triangle
        ]
        coordinates {(8.73, 9.74)};
        \addlegendentry{rTrie v5}
        
        \addplot+[
            mark size = 4pt,
            only marks,
            color=blue,
            mark=triangle
        ]
        coordinates {(9.16, 8.13)};
        \addlegendentry{rTrie IL 512}
        
        \addplot+[
            mark size = 4pt,
            only marks,
            color=teal,
            mark=triangle
        ]
        coordinates {(9.95, 6.09)};
        \addlegendentry{rTrie v}

        \addplot+[
            mark size = 4pt,
            only marks,
            color=red,
            mark=triangle,
            every mark/.append style={rotate=180}
        ]
        coordinates {(8.77, 8.72)};
        \addlegendentry{trie v5}
        \addplot+[
            mark size = 4pt,
            only marks,
            color=blue,
            mark=triangle,
            every mark/.append style={rotate=180}
        ]
        coordinates {(9.30, 7.46)};
        
        \addlegendentry{trie IL 512}
        \addplot+[
            mark size = 4pt,
            only marks,
            color=teal,
            mark=triangle,
            every mark/.append style={rotate=180}
        ]
        coordinates {(9.99, 5.21)};
        \addlegendentry{trie v}
        \end{axis}
    \end{tikzpicture}

    \caption{Space vs.~time trade-off for all alternative tested on the 3 datasets.}
    \label{fig:experimental-query-time}
\end{figure}
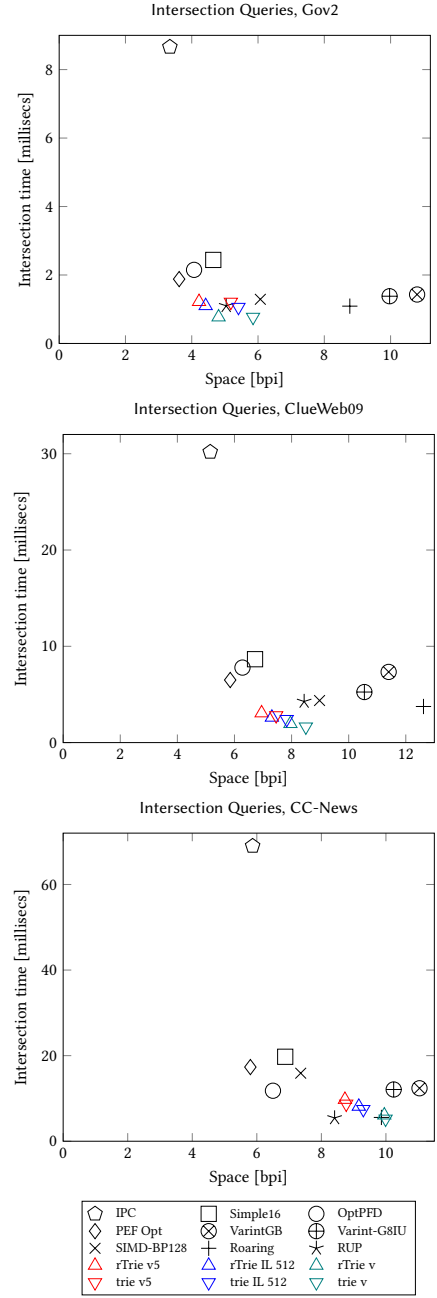
As it can be seen, our approaches introduce relevant trade-offs. 
We compare next with the most competitive approaches on the different datasets. 
For \textsf{Gov2}, \texttt{rTrie} uses 1.166--1.329 times the space of \texttt{PEF}, the former being 1.549--2.442 times faster.
\texttt{rTrie} uses 0.481--0.548 times the space of \texttt{Roaring}, the former being 
up to 1.415 times faster. Finally, \texttt{rTrie} uses 0.837--0.954 times the space of \texttt{RUP}, the former being 
up to 1.428 times faster.
For \textsf{ClueWeb09}, \texttt{rTrie} uses 1.188--1.361 times the space of \texttt{PEF}, the former being 2.117--3.316 times faster. Also, \texttt{rTrie} uses 0.551--0.631 times the space of \texttt{Roaring}, the former being 1.221--1.913 times faster.
Finally, \texttt{rTrie} uses 0.823--0.943 times the space of \texttt{RUP}, the former being 1.391--2.178 times faster.
For \textsf{CC-News}, the resulting inverted index has considerably less runs in the inverted lists, hence the space usage of \texttt{trie} and \texttt{rTrie} is about the same. However, \texttt{trie} is faster than \texttt{rTrie}, as the code to handle runs introduces an overhead that does not pay off in this case. 
So, we will use \texttt{trie} to compare here. It uses 1.512--1.722 times the space of \texttt{PEF}, the former being 1.987--3.326 times faster.
\texttt{trie} uses 0.889--1.013 times the space of \texttt{Roaring}, the former being 
up to 1.067 times faster.
Finally, \texttt{trie} uses 1.043--1.188 times the space of \texttt{RUP}, the former being 
up to 1.044 times faster.

We can conclude that in all tested datasets, at least one of our trade-offs is the fastest, outperforming the highly-engineered ultra-efficient set compression techniques we tested. 

%% file: 06-conclusions.tex
\section{Conclusions} \label{sec:conclusions}

We conclude that Trabb-Pardo's intersection algorithm \cite{TrabbPardo} implemented using compact binary tries yields a solution to the offline set intersection problem that is appealing both in theory and practice. Namely, our proposal has: (1)  theoretical guarantees of compressed space usage, (2) adaptive intersection computation time, and (3) highly competitive practical performance. 
Regarding our experimental results, we show that our algorithm computes intersections 1.55--3.32 times faster than highly-competitive Partitioned Elias-Fano indexes \cite{OVsigir14}, using 1.15--1.72 times their space. We also compare with Roaring Bitmaps \cite{LKKDOSKspe18}, being 1.07--1.91 times faster while using about 0.48--1.01 times their space. Finally, we compared to RUP \cite{Pdcc21}, being 1.04--2.18 times faster while using 0.82--1.19 times its space. 

After this work, multiple avenues for future research are open. For instance, novel data structures supporting operations $\rank_{\bit{1}}$ and $\select_{\bit{1}}$ have emerged recently \cite{Kspire22}. These offer interesting trade-offs, using much less space than then ones we used in our implementation, with competitive operation times. We think these can improve our space usage significantly, keeping our competitive query times. Another interesting line is that of alternative  binary trie compact representations. E.g., a DFS representation (rather than BFS, as the one used in this paper), which would potentially reduce the number of cache misses when traversing the tries. Finally, our representation would support dynamic sets (where insertion and deletion of elements are allowed) if we use dynamic binary tries \cite{ADRalgo16}